\documentclass{amsart}
\usepackage{geometry}
\usepackage{multicol}
\usepackage{fullpage,graphicx,subfigure,mathpazo,color}
\usepackage{amsthm,amsmath,amscd,tikz,mathrsfs,amssymb}
\usepackage[normalem]{ulem}
\usepackage{setspace}
\usepackage{graphicx}
\usepackage{float}
\usepackage{indentfirst}
\numberwithin{equation}{section}
\newcommand{\ii}{\mathrm{i}}
\newcommand{\ee}{\mathrm{e}}
\newcommand{\dd}{\mathrm{d}}

\newtheorem{theorem}{Theorem}

\newtheorem{prop}{Proposition}

\usepackage{graphicx}      
\usepackage{titlesec}

\titleformat{\section}{\centering\LARGE\bfseries}{\thesection}{1em}{}
\titleformat{\subsection}{\Large\bfseries}{\thesubsection}{1em}{}
\numberwithin{figure}{section}

\begin{document}
	\title{Nondegenerate solitons in the integrable fractional coupled Hirota equation}
	\author{Ling An}
	\address{School of Mathematics, South China University of Technology, Guangzhou, China, 510641}
	\email{maal@mail.scut.edu.cn}
	
	\author{Liming Ling}
	\address{School of Mathematics, South China University of Technology, Guangzhou, China, 510641}
	\email{linglm@scut.edu.cn}
	
	\author{Xiaoen Zhang}
	\address{School of Mathematics, South China University of Technology, Guangzhou, China, 510641}
	\email{zhangxiaoen@scut.edu.cn}
	
	\begin{abstract}
		In this paper, based on the nonlinear fractional equations proposed by Ablowitz, Been, and Carr in the sense of Riesz fractional derivative, we explore the fractional coupled Hirota equation and give its explicit form. Unlike the previous nonlinear fractional equations, this type of nonlinear fractional equation is integrable. Therefore, we obtain the fractional $n$-soliton solutions of the fractional coupled Hirota equation by inverse scattering transformation in the reflectionless case. In particular, we analyze the one- and two-soliton solutions of the fractional coupled Hirota equation and prove that the fractional two-soliton can also be regarded as a linear superposition of two fractional single solitons as $|t|\to\infty$. Moreover, under some special constraint,
		we also obtain the nondegenerate fractional soliton solutions and give a simple analysis for them.\\
		{\bf Keywords:\ }fractional coupled Hirota equation, anomalous dispersive relation, inverse scattering transform, squared eigenfunction, nondegenerate soliton solution.
	\end{abstract}
	
	\maketitle
	\section{Introduction}
	Fractional calculus belongs to the field of mathematical analysis, which can be used to study the application of integrals and derivatives of any order. The fractional calculus first emerged from some speculations made by Leibniz ($1695,\ 1697$) and Euler ($1730$) and has been developed up to nowadays. Fractional differential equations have profound physical background and rich theoretical connotation, and have been well applied in many fields, such as viscoelastic hydrodynamics \cite{grigorenko2003chaotic,hayat2004periodic,tripathi2010peristaltic}, signal processing \cite{lohmann1996some,sejdic2011fractional}, material anomalous diffusion and heat conduction \cite{metzler1999anomalous,henry2006anomalous}, biology \cite{ahmed2007fractional,el2009exact,liu2011novel}, and so on. With the development of research, people began to explore the physical and mechanical background of the fractional operators and found that many classical mechanical theories defined by integer order differential operators were challenged. For example, because the turbulent velocity field of the atmosphere will change its vibration and direction randomly and violently, the velocity of the turbulent velocity field does not meet the requirements of classical mechanics with integer order differentiability. In this case, the non-classical mechanics based on fractional calculus can describe these ``abnormal" physical phenomena well. Against that backdrop, many fractional differential equation models have been established.
	
	So far, there are many types of the fractional equations. For example, in \cite{al2018high,li2021symmetry}, different versions of the fractional nonlinear Schr\"{o}dinger (NLS) equation have been studied, and soliton type solutions have been found. However, none of these versions of the fractional NLS equation is integrable in the inverse scattering transform (IST) sense. In $2022$, Ablowitz, Been, and Carr proposed the integrable fractional nonlinear soliton equations based on the Riesz fractional derivative and IST integrability, such as the fractional NLS equation, the fractional Korteweg-deVries (KdV) equation, the fractional modified KdV (mKdV) equation, the fractional sine-Gordon  equation and the fractional sinh-Gordon  equation \cite{ablowitz2022fractional,ablowitz2022integrable}. This research process has three key elements: a general evolution equation solvable by the inverse scattering transformation, an anomalous dispersion relation of the fractional equation, and the completeness of the square eigenfunction. Subsequently, fractional higher-order mKdV equation \cite{zhang2022interactions}, fractional higher-order NLS equation \cite{weng2022dynamics}, and so on were studied \cite{yan2022new}. It can be found that all of these above researches
	are related to the $2\times2$ matrix scattering problem, but to our best knowledge, there is little study for
	higher-order matrix scattering problem. In this paper, we expand the fractional integrable equation to
	the $3\times3$ Lax equation and give the fractional coupled Hirota (fcHirota) equation. Then we construct
	the soliton type solutions and give a simple asymptotic analysis.
	
	The NLS equation can describe electromagnetic waves with slowly varying amplitude in nonlinear media \cite{hasegawa1973transmission}. In order to increase the bit rates, the pulse width needs to be reduced. However, the NLS equation becomes insufficient when the pulse length becomes comparable to the wavelength.  Then the additional terms must be added, and the resulting pulse propagation is called the higher-order NLS equation \cite{mollenauer1980experimental}, which includes effects like third-order dispersion, self-steepening, and stimulated Raman scattering. In $1974$, Manakov proposed the coupled NLS equation  \cite{manakov1974theory}. In $1992$, Tasgal and Potasek proposed the coupled higher-order NLS equation, which is the coupled version of the  Hirota equation \cite{tasgal1992soliton}. The coupled Hirota equation has the form 
	\begin{equation}\label{3-cHirota}
		\begin{cases}
			\ii u_{t}+\frac{1}{2}u_{xx}+u(|u|^{2}+|v|^{2})+\ii\varepsilon\big(u_{xxx}+u_{x}(6|u|^{2}+3|v|^{2})+3uv^{*}v_{x}\big)=0,\\[2pt]
			\ii v_{t}\ +\frac{1}{2}v_{xx}+v(|u|^{2}+|v|^{2})+\ii\varepsilon\big(v_{xxx}+v_{x}(6|v|^{2}+3|u|^{2})+3vu^{*}u_{x}\big)=0,
		\end{cases}
	\end{equation}
	where $u(x,t)$ and $v(x,t)$ represent the slowly varying complex amplitudes in two interacting optical modes, $^{*}$ means the complex conjugate, and $\varepsilon$ is a small dimensionless real parameter. When $\varepsilon=0$, the equation \eqref{3-cHirota} is reduced to the coupled NLS equation. The equation \eqref{3-cHirota} can not only simulate the propagation of ultrashort optical pulses in birefringent fiber \cite{porsezian1997optical}, but also describe the collision of two waves caused by severe weather in deep ocean \cite{ankiewicz2010rogue}. Based on the reported results of the equation \eqref{3-cHirota}, such as the Lax pair, the classical Darboux transformation, and so on \cite{bindu2001dark,wang2014generalized}, we explore the fractional form of the coupled Hirota equation.
	
	It should be pointed out that for the coupled Hirota equation, there exist many types of vector solitons, such as bright-bright solitons, bright-dark solitons, dark-dark solitons and so on \cite{wang2014integrability}. From the study on vector solitons of the coupled Hirota equation, we can confirm that there exists a degeneracy in the structure of bright solitons, that is, the solitons in the multi-mode fibers propagate with the same wave numbers. To avoid the degeneracy, different propagation constants can be appropriately introduced in the structure of fundamental bright solitons of the coupled Hirota equation. Thus, a new class of fundamental bright solitons, namely nondegenerate fundamental vector bright solitons, can be obtained. The nondegenerate soliton solution of nonlinear integrable equations is a very interesting and important issue, which had been explored by many scholars recently \cite{qin2019nondegenerate,stalin2021nondegenerate,stalin2022dynamics}. We argue that this phenomenon is equally presented in the fcHirota equation.
	
	The organization of this work is as follows: In Sec.\ref{3-sec.2}, we first define the fcHirota equation from a $3\times 3$ scattering problem by introducing a recursive operator $\mathcal{L}$. According to the linearized equation of the fcHirota equation, the anomalous dispersion relation of the fcHirota equation is deduced. Then we establish a matrix Riemann-Hilbert problem (RHP) for the fcHirota equation by assuming that the potential functions are sufficiently smooth and rapidly tend to zero as $|x|\rightarrow\infty$, and give a constraint condition of the matrix ${\bf V}(\lambda;x,t)$ as $|x|\rightarrow\infty$. Subsequently, we solve the RHP. At the same time, we construct the square eigenfunctions and the adjoint square eigenfunctions, and then use their completeness to derive the explicit form of the fcHirota equation. In Sec.\ref{3-sec.3}, we explore the fractional $n$-soliton solutions and give a simple analysis for the fractional one- and two-soliton solutions. In addition, we investigate the nondegenerate case of the  fractional soliton solutions of the fcHirota equation.

	\section{The IST for the fcHirota equation}\label{3-sec.2}
	In this section, we would like to combine the idea in \cite{ablowitz2022fractional} to derive the integrable fcHirota equation and solve this equation by the IST with the RHP. 
	
	\subsection{The fcHirota equation and anomalous dispersive relation}
	Firstly, we consider the Ablowitz-Kaup-Newell-Segur (AKNS) scattering problem
	\begin{subequations}\label{3-Lax-Phi}
		\begin{equation}\label{3-Lax-Phi-a}
			\Phi_{x}={\bf U}\Phi,
		\end{equation}
		\begin{equation}\label{3-Lax-Phi-b}
			\Phi_{t}={\bf V}\Phi,
		\end{equation}
	\end{subequations}
	where $\Phi(\lambda;x,t)$ is a wave function associated with $\lambda$,  $\lambda\in\mathbb{C}$ is an eigenvalue parameter, and ${\bf U}(\lambda;x,t)$ can be expressed in the following polynomial with respect to $\lambda$
	\begin{equation*}
		\begin{split}
			&{\bf U}=-\ii\lambda\Lambda+{\bf U}_{0},\ \ \ \ \Lambda={\rm diag}\left(1,-1,-1\right),\ \ \ 
			{\bf U}_{0}(x,t)=\begin{bmatrix}
				0 & {\bf q}(x,t)\\[4pt]
				{\bf r}^{\top}(x,t) &0\\
			\end{bmatrix},\\
			&{\bf q}(x,t)=
			\begin{bmatrix}
				q_{1}(x,t)&q_{2}(x,t)
			\end{bmatrix},\ \ \ \ {\bf r}(x,t)=\begin{bmatrix}
				r_{1}(x,t)&r_{2}(x,t)
			\end{bmatrix},
		\end{split}
	\end{equation*}
	the superscript $^{\top}$ denotes the transpose, and 
	\begin{equation*}
		{\bf V}(\lambda;x,t)=\begin{bmatrix}
			{\bf V}_{1}(\lambda;x,t) &{\bf V}_{2}(\lambda;x,t)\\[4pt]
			{\bf V}_{3}^{\top}(\lambda;x,t) &{\bf V}_{4}(\lambda;x,t)
		\end{bmatrix}.
	\end{equation*}
	Note that the matrix ${\bf V}(\lambda;x,t)$ has the same block structure as the matrix ${\bf U}(\lambda;x,t)$. In order to ensure the compatibility of system \eqref{3-Lax-Phi}, the following zero curvature equation needs to be
	satisfied
	\begin{equation}\label{3-ZCE}
		{\bf U}_{t}-{\bf V}_{x}+[{\bf U},{\bf V}]=0,
	\end{equation}
	where the commutator is defined by $[{\bf U},{\bf V}]\equiv {\bf UV-VU}.$
	
	Combining the equation \eqref{3-ZCE} and making appropriate assumptions, the integrable hierarchy associated with the third-order scattering problem can be obtained
	\begin{equation}\label{3-three-scattering}
		\ii
		\begin{bmatrix}
			{\bf q}\\[4pt]
			-{\bf r}
		\end{bmatrix}_{t}=\mathcal{F}(\mathcal{L})
		\begin{bmatrix}
			{\bf q}\\[4pt]
			{\bf r}
		\end{bmatrix},
	\end{equation}
	where $\mathcal{F}(\mathcal{L})$ is a function with respect to the operator $\mathcal{L}$, and here $\mathcal{F}(\mathcal{L})=\mathcal{L}^{n}\ (n=1,2,\cdots)$. The recursion operator $\mathcal{L}$ is
	\begin{equation}
		\begin{split}
			&\mathcal{L}\begin{bmatrix}
				\widetilde{{\bf q}}\\[4pt]
				\widetilde{{\bf r}}
			\end{bmatrix}=\ii
			\begin{bmatrix}
				\widetilde{{\bf q}}_{x}-{\bf q}\partial_{-}^{-1}({\bf r}^{\top}\widetilde{{\bf q}}-\widetilde{{\bf r}}^{\top}{\bf q})-\big(\partial_{-}^{-1}(\widetilde{{\bf q}}{\bf r}^{\top}+{\bf q}\widetilde{{\bf r}}^{\top})\big){\bf q}\\[4pt]
				-\widetilde{{\bf r}}_{x}-{\bf r}\partial_{-}^{-1}(\widetilde{{\bf q}}^{\top}{\bf r}-{\bf q}^{\top}\widetilde{{\bf r}})+\big(\partial_{-}^{-1}({\bf r}\widetilde{{\bf q}}^{\top}+\widetilde{{\bf r}}{\bf q}^{\top})\big){\bf r}
			\end{bmatrix},\ \ \ \ \ \ \ 
			\begin{split}
				\widetilde{{\bf q}}=&\begin{bmatrix}
					\widetilde{q}_{1}&\widetilde{q}_{2}
				\end{bmatrix},\\[4pt]
				\widetilde{{\bf r}}=&\begin{bmatrix}
					\widetilde{r}_{1}&\ \widetilde{r}_{2}
				\end{bmatrix},
			\end{split}
		\end{split}
	\end{equation}
	where $\partial_{-}^{-1}=\int_{-\infty}^{x}\dd y$. The adjoint operator of $\mathcal{L}$ is defined as
	\begin{equation}
		\mathcal{L}^{A}
		\begin{bmatrix}
			\widetilde{{\bf q}}\\[4pt]
			\widetilde{{\bf r}}
		\end{bmatrix}=\ii
		\begin{bmatrix}
			-\widetilde{{\bf q}}_{x}-{\bf r}\partial_{+}^{-1}({\bf q}^{\top}\widetilde{{\bf q}}+\widetilde{{\bf r}}^{\top}{\bf r})-\big(\partial_{+}^{-1}(\widetilde{{\bf q}}{\bf q}^{\top}+{\bf r}\widetilde{{\bf r}}^{\top})\big){\bf r}\\[4pt]
			\widetilde{{\bf r}}_{x}+{\bf q}\partial_{+}^{-1}(\widetilde{{\bf q}}^{\top}{\bf q}+{\bf r}^{\top}\widetilde{{\bf r}})-\big(\partial_{+}^{-1}({\bf q}\widetilde{{\bf q}}^{\top}+\widetilde{{\bf r}}{\bf r}^{\top})\big){\bf q}
		\end{bmatrix},
	\end{equation}
	where $\partial_{+}^{-1}=\int_{x}^{+\infty}\dd y$.
	In this case, given the value of $n$ and the symmetry relation between ${\bf q}(x,t)$ and ${\bf r}(x,t)$, the corresponding integer order integrable equation can be obtained. In order to introduce the fractional integrable equations, we need to make some appropriate changes to the function $\mathcal{F}(\mathcal{L})$. That is to say that we need to add the fractional power of the operator $\mathcal{L}$ to $\mathcal{F}(\mathcal{L})$. In this paper, we take $\mathcal{F}(\mathcal{L})=(\beta+\alpha\mathcal{L})\mathcal{L}^{2}|\mathcal{L}^{2}|^{\epsilon}$ with $\epsilon\in[0,1)$, $\alpha,\beta\in\mathbb{R}$, then
	\begin{equation}\label{3-fcHirota-1}
		\begin{bmatrix}
			{\bf q}\\
			{\bf r}
		\end{bmatrix}_{t}=\sigma_{3}|\mathcal{L}^{2}|^{\epsilon}
		\begin{bmatrix}
			\ii\beta({\bf q}_{xx}-2{\bf q}{\bf r}^{\top}{\bf q})-\alpha({\bf q}_{xxx}-3{\bf q}_{x}{\bf r}^{\top}{\bf q}-3{\bf q}{\bf r}^{\top}{\bf q}_{x})\\[4pt]
			\ii\beta({\bf r}_{xx}-2{\bf r}{\bf q}^{\top}{\bf r})+\alpha({\bf r}_{xxx}-3{\bf r}_{x}{\bf q}^{\top}{\bf r}-3{\bf r}{\bf q}^{\top}{\bf r}_{x})
		\end{bmatrix},
	\end{equation}
	where $\sigma_{3}={\rm diag}(1,-1)$. 
	In particular, when $\beta=0$, ${\bf r}^{\top}=\sigma{\bf q}^{\dagger},\ \sigma=\pm 1$, we can derive the fractional coupled complex mKdV  equation, and when $\alpha=0$, ${\bf r}^{\top}=\sigma{\bf q}^{\dagger},\ \sigma=\pm 1$, we can derive the fractional coupled NLS equation. Without loss of generality, we take $\beta=\frac{1}{2}$, ${\bf r}^{\top}=-{\bf q}^{\dagger}$, then we can get the  fcHirota equation which can be regarded as the fractional type to the coupled Hirota equation \eqref{3-cHirota}. Under this case, the corresponding operator $\mathcal{F}(\mathcal{L})$ becomes $\mathcal{F}_{h}(\mathcal{L})=(\frac{1}{2}+\alpha\mathcal{L})\mathcal{L}^{2}|\mathcal{L}^{2}|^{\epsilon}$. The explicit form of the fcHirota equation will be given in subsection \ref{sec.2.3}.
	
	Considering the linearization of the fcHirota equation \eqref{3-fcHirota-1} (i.e. when $\beta=\frac{1}{2},\ {\bf r}^{\top}=-{\bf q}^{\dagger}$)
	\begin{equation}\label{3-fcHirota-1-linera} 
		\ii{\bf q}_{t}+|-\partial^{2}|^{\epsilon}\Big(\frac{1}{2}{\bf q}_{xx}+\ii\alpha{\bf q}_{xxx}\Big)=0,
	\end{equation}
	where $|-\partial^{2}|^{\epsilon}$ is the Riesz fractional derivative \cite{ablowitz2022fractional}.
	We substitute the formal solution
	\begin{equation}
		{\bf q}=\begin{bmatrix}
			q_{1}&q_{2}
		\end{bmatrix}\thicksim
		\begin{bmatrix}
			\ee^{\ii\big(k_{1} x-\omega_{1}(k_{1})t\big)}&\ee^{\ii\big(k_{2} x-\omega_{2}(k_{2})t\big)}
		\end{bmatrix}
	\end{equation}
	into the equation \eqref{3-fcHirota-1-linera}, then the anomalous dispersive relations can be obtained 
	\begin{equation}
		\omega_{j}(k_{j})=\mathcal{F}_{h}(-k_{j})=(\frac{1}{2}-\alpha k_{j})k_{j}^{2}|k_{j}^{2}|^{\epsilon},\ \ \ j=1,2,
	\end{equation}
	in which we have $\mathcal{F}_{h}(k_{j})=\omega_{j}(-k_{j})$.
	
	\subsection{The direct scattering}	
	We assume that the potential ${\bf q}(x,t)$ is sufficiently smooth and rapidly tends to zero as $|x|\to\infty$. In integer order integrable equations, the matrix ${\bf V}(\lambda;x,t)$ can be expressed as a matrix function concerning $\lambda$, while for the fractional  integrable equations in the sense of this paper, the matrix ${\bf V}(\lambda;x,t)$ usually cannot be given precisely. Nevertheless, here we need to give the constraints \cite{weng2022dynamics}
	\begin{equation}
		\begin{cases}
			{\bf V}_{j}(\lambda;x,t)\to 0,\ \ j=1,2,3,\\[4pt]
			{\bf V}_{4}(\lambda;x,t)\to \ii\mathcal{F}_{h}(2\lambda)\mathbb{I}_{2},
		\end{cases}
		\ \ |x|\to\infty.
	\end{equation}
	Therefore, we can derive the asymptotic behavior
	\begin{equation}\label{3-Phi-asy-be}
		\Phi^{\pm}(\lambda;x,t)\sim \ee^{-\ii\lambda\Lambda x+\ii\mathcal{F}_{h}(2\lambda)\Omega t},\ \ |x|\to\infty,
	\end{equation}
	where $\Omega={\rm diag}(0,1,1)$,  $\Phi^{\pm}(\lambda;x,t)=
	\begin{bmatrix}
		\phi^{\pm}_{1},&\phi^{\pm}_{2},&\phi^{\pm}_{3}
	\end{bmatrix}$, and the superscripts $^{\pm}$ refer to the cases $x\to\pm\infty$. Then the modified Jost solutions $\Psi^{\pm}(\lambda;x,t)$ can be defined as
	\begin{equation}\label{3-Psi-def}
		\Psi^{\pm}(\lambda;x,t)=\Phi^{\pm}(\lambda;x,t) \ee^{\ii\lambda\Lambda x-\ii\mathcal{F}_{h}(2\lambda)\Omega t}\rightarrow \mathbb{I}_{3},\ \ \ |x|\to\infty.
	\end{equation}
	And we can get an equation of $\Psi(\lambda;x,t)$, which is equivalent to the equation \eqref{3-Lax-Phi-a}
	\begin{equation}\label{3-Psi}
		\Psi_{x}=-\ii\lambda[\Lambda,\Psi]+{\bf U}_{0}\Psi.
	\end{equation}
	In addition, the Volterra integral equations for $\Psi^{\pm}(\lambda;x,t)$ can be given as
	\begin{equation}\label{3-Psi-Volterra}
		\Psi^{\pm}(\lambda;x)=\mathbb{I}_{3}+\int_{\pm\infty}^{x}\ee^{-\ii\lambda{\rm ad}\widehat{\Lambda}(x-y)}{\bf U}_{0}(y)\Psi^{\pm}(\lambda;y)dy,
	\end{equation}
	where $\ee^{{\rm ad}\widehat{\Lambda}}{\bf A}=\ee^{\Lambda}{\bf A}\ee^{-\Lambda},$ $\Psi^{\pm}(\lambda;x,t)=\begin{bmatrix}
		\psi^{\pm}_{1},&\psi^{\pm}_{2},&\psi^{\pm}_{3}
	\end{bmatrix}$.
	Let $\mathbb{C^{+}}=\{\lambda|\Im\lambda>0\},\ \mathbb{C^{-}}=\{\lambda|\Im\lambda<0\}.$  Based on the equation \eqref{3-Psi-Volterra}, we can find that the first column of $\Psi^{+}$ (i.e., $\psi^{+}_{1}$) only contains the exponential factor $\ee^{2\ii\lambda(x-y)}$, which decays when $\lambda\in\mathbb{C^{-}}$ due to the condition $y>x$ in the integral. Similarly, the second and the third columns (i.e., $\psi^{+}_{2},\ \psi^{+}_{3}$) decay in the region $\lambda\in\mathbb{C^{+}}$. We can analyze the analyticity of $\Psi^{-}$ in the same way, and then we summarize the above analysis as follows:  
	for any $x,t\in\mathbb{R}$, $\psi_{1}^{-},\ \psi_{2}^{+},\ \psi_{3}^{+}$ are analytic for $\lambda\in\mathbb{C}^{+}$, while $\psi_{1}^{+},\ \psi_{2}^{-},\ \psi_{3}^{-}$ are analytic for $\lambda\in\mathbb{C}^{-}$.
	
	From above, we know $\Phi^{\pm}(\lambda;x,t)$ can be regarded as the fundamental solution of the equation \eqref{3-Lax-Phi}. By the theory of ordinary differential equations, there exists a matrix ${\bf S}(\lambda)=\left(s_{ij}(\lambda)\right)_{i,j=1,2,3}$ between them obeying the relation
	\begin{equation}\label{3-Phi-S}
		\Phi^{-}=\Phi^{+}{\bf S}(\lambda),\ \ \ \ \lambda\in\mathbb{R},
	\end{equation}
	and the matrix ${\bf S}(\lambda)$ is called the scattering matrix. 
	Combined with the equation \eqref{3-Psi}, the equation \eqref{3-Phi-S} can be rewritten as
	\begin{equation}\label{3-Psi-S}
		\Psi^{-}=\Psi^{+}\ee^{-\ii\lambda{\rm ad}\widehat{\Lambda} x+\ii\mathcal{F}_{h}(2\lambda){\rm ad}\widehat{\Omega} t}{\bf S}(\lambda),\ \ \ \lambda\in\mathbb{R},
	\end{equation}
	where $\ee^{{\rm ad}\widehat{\Omega}}{\bf A}=\ee^{\Omega}{\bf A}\ee^{-\Omega}.$ 
	Now we assume $\det\Phi^{\pm}=\mu(\lambda)$ which implies $\det {\bf S}(\lambda)=1$. Using \eqref{3-Phi-S},
	\begin{equation}
		s_{1j}(\lambda)=\frac{1}{\mu}\big|\phi_{j}^{-},\phi_{2}^{+},\phi_{3}^{+}\big|,\ \ s_{2j}(\lambda)=\frac{1}{\mu}\big|\phi_{1}^{+},\phi_{j}^{-},\phi_{3}^{+}\big|,\ \ s_{3j}(\lambda)=\frac{1}{\mu}\big|\phi_{j}^{-},\phi_{1}^{+},\phi_{2}^{+}\big|,\ \ \ \ j=1,2,3,
	\end{equation}
	and we can find that $s_{11}(\lambda)$ is analytic in $\mathbb{C^{+}}$, $\{s_{jk}(\lambda)\}_{j,k=2,3}$ are analytic in $\mathbb{C^{-}}$.  In general, $s_{12}(\lambda),\ s_{13}(\lambda)$, $s_{21}(\lambda),\ s_{31}(\lambda)$ cannot be extended off the
	real $\lambda$-axis.

	\subsection{The inverse scattering}
	In what follows, we shall formulate the RHP using the properties of $\Psi^{\pm}(\lambda;x,t)$. To
	this end, we combine the ideas in \cite{yang2010nonlinear} to construct two matrix functions ${\bf P}^{\pm}(\lambda;x,t)$ which are analytic in $\mathbb{C}^{\pm}$, respectively,
	\begin{equation}
		\begin{split}
			&{\bf P}^{+}(\lambda;x,t)=\begin{bmatrix}
				\psi_{1}^{-}&\psi_{2}^{+}&\psi_{3}^{+}
			\end{bmatrix}=\Psi^{-}{\bf H}_{1}+\Psi^{+}{\bf H}_{2},\\
			&{\bf P}^{-}(\lambda;x,t)=\begin{bmatrix}
				\widetilde{\psi}^{-}_{1}&\widetilde{\psi}^{+}_{2}&\widetilde{\psi}^{+}_{3}
			\end{bmatrix}^{\top}={\bf H}_{1}\widetilde{\Psi}^{-}+{\bf H}_{2}\widetilde{\Psi}^{+},
		\end{split}
	\end{equation}
	where 
	\begin{equation}
		{\bf H}_{1}={\rm diag}(1,0,0),\ \ \ \ {\bf H}_{2}={\rm diag}(0,1,1),\ \ \ 
		\widetilde{\Psi}^{\pm}=(\Psi^{\pm})^{-1}=\begin{bmatrix}
			\widetilde{\psi}^{\pm}_{1}&\widetilde{\psi}^{\pm}_{2}&\widetilde{\psi}^{\pm}_{3}
		\end{bmatrix}^{\top}.
	\end{equation}
	There are some properties of ${\bf P}^{\pm}(\lambda;x,t)$ which can be derived through a series of complex analyses,
	and we summarize them in the proposition \ref{3-prop-1}.
	\begin{prop}\label{3-prop-1}
		The matrix functions ${\bf P}^{\pm}(\lambda;x,t)$ have the following properties:\\
		$(1)$\ $\det {\bf P}^{+}(\lambda;x,t)=s_{11}(\lambda),\ \det {\bf P}^{-}(\lambda;x,t)=\widehat{s}_{11}(\lambda),$ where ${\bf S}^{-1}(\lambda)=\big(\widehat{s}_{jk}(\lambda)\big)_{j,k=1,2,3}$.\\[4pt]
		$(2)$\ $\lim\limits_{x\to+\infty}{\bf P}^{+}(\lambda;x,t)=\begin{bmatrix}
			s_{11}(\lambda)&0\\[4pt]
			0&\mathbb{I}_{2}
		\end{bmatrix},\ \lambda\in\mathbb{C^{+}},$\ \ \ \ $\lim\limits_{x\to-\infty}{\bf P}^{-}(\lambda;x,t)=\begin{bmatrix}
			\widehat{s}_{11}(\lambda)&0\\[4pt]
			0&\mathbb{I}_{2}
		\end{bmatrix},\ \lambda\in\mathbb{C^{-}}.$\\[4pt]
		$(3)$\ ${\bf P}^{\pm}(\lambda;x,t)\to\mathbb{I}_{3},\ \ \lambda\in\mathbb{C}^{\pm}$.
	\end{prop}
	Moreover, from the fact that ${\bf U}_{0}^{\dagger}(x,t)=-{\bf U}_{0}(x,t)$, we have $\Psi^{\dagger}(\lambda^{*};x,t)=\Psi^{-1}(\lambda;x,t)$. Then the symmetry relation of ${\bf S}(\lambda)$ can be obtained by combining the equation \eqref{3-Psi-S}
	\begin{equation}\label{3-sym-S}
		{\bf S}^{\dagger}(\lambda^{*})={\bf S}^{-1}(\lambda),
	\end{equation}
	which implies $s^{*}_{11}(\lambda^{*})=\widehat{s}_{11}(\lambda),\ \lambda\in\mathbb{C^{-}}$.	Note that the inverse of the solution of the adjoint matrix spectral problem is exactly the solution of the matrix spectral problem. So we can define a sectional analytic solution ${\bf M}(\lambda;x,t)$
	\begin{equation}
		\begin{split}
			&{\bf M}^{+}(\lambda;x,t)={\bf P}^{+}(\lambda;x,t),\ \ \ \ \ \ \ \ \ \  \lambda\in\mathbb{C^{+}},\\
			&{\bf M}^{-}(\lambda;x,t)=\big({\bf P}^{-}(\lambda;x,t)\big)^{-1},\ \ \  \lambda\in\mathbb{C^{-}}.
		\end{split}
	\end{equation}
	Then we can formulate the following RHP.\\
	{\bf Riemann-Hilbert problem $1$.\ }We can find a matrix ${\bf M}(\lambda;x,t)$ with the following properties:
	\begin{itemize}
		\item ${\bf Analyticity:}$\ ${\bf M}^{\pm}(\lambda;x,t)$ are analytic in $\mathbb{C}^{\pm}$, respectively.
		\item ${\bf Jump\ condition:}$\ The two matrix functions ${\bf M}^{\pm}(\lambda;x,t)$ can be linked by
		\begin{equation*}
			{\bf M}^{+}(\lambda;x,t)={\bf M}^{-}(\lambda;x,t){\bf J}(\lambda;x,t),
		\end{equation*}
		where the matrix ${\bf J}(\lambda;x,t)$ is
		\begin{equation*}
			\begin{split}
				{\bf J}(\lambda;x,t)&=\ee^{-\ii\lambda{\rm ad}\widehat{\Lambda} x+\ii\mathcal{F}_{h}(2\lambda){\rm ad}\widehat{\Omega} t}\big(({\bf H}_{1}{\bf S}^{-1}+{\bf H}_{2})({\bf S}{\bf H}_{1}+{\bf H}_{2})\big)\\[4pt]
				&=\begin{bmatrix}
					1&\ee^{-\ii(2\lambda x+\mathcal{F}_{h}(2\lambda)t)}\widehat{s}_{12}&\ee^{-\ii(2\lambda x+\mathcal{F}_{h}(2\lambda)t)}\widehat{s}_{13}\\[4pt]
					\ee^{\ii(2\lambda x+\mathcal{F}_{h}(2\lambda)t)}s_{21}&1&0\\[4pt]
					\ee^{\ii(2\lambda x+\mathcal{F}_{h}(2\lambda)t)}s_{31}&0&1
				\end{bmatrix}.
			\end{split}
		\end{equation*}
		\item ${\bf Normalization:}$\ ${\bf M}(\lambda;x,t)\rightarrow\mathbb{I}_{3}$, as $\lambda\to \infty$.
	\end{itemize}
	
	The RHP with zeros generates soliton solutions. The uniqueness of the RHP $1$ does not hold unless the zeros of $\det {\bf P}^{\pm}(\lambda;x,t)$ in $\mathbb{C}^{\pm}$ are specified and the kernel structures of ${\bf P}^{\pm}(\lambda;x,t)$ at these zeros are determined. In view of \eqref{3-sym-S}, it is clear that if $\lambda$ is a zero of $s_{11}(\lambda)$, $\lambda^{*}$ will be a zero of $\widehat{s}_{11}(\lambda)$. Based on these results, we suppose that $s_{11}(\lambda)$ has $n$ simple zeros $\{\lambda_{k}\}_{1}^{n}$. Accordingly, each zero of $\widehat{s}_{11}(\lambda)$ is complex conjugate to that
	of $s_{11}(\lambda)$, and we define them by $\{\widehat{\lambda}_{j}\}_{1}^{n}$. Under this assumption, the kernel of ${\bf P}^{+}(\lambda;x,t)$ contains a single column vector $|v_{k}\rangle=\begin{bmatrix}
		v_{1k},&v_{2k},&v_{3k}
	\end{bmatrix}^{\top}$ and the kernel of ${\bf P}^{-}(\lambda;x,t)$ contains a single row vector $\langle\widehat{v}_{j}|=\begin{bmatrix}
		\widehat{v}_{j1},&\widehat{v}_{j2},&\widehat{v}_{j3}
	\end{bmatrix}$,
	\begin{equation}\label{3-P-kernel}
		{\bf P}^{+}(\lambda_{k};x,t)|v_{k}\rangle=0,\ \ \ \ \langle\widehat{v}_{j}|{\bf P}^{-}(\widehat{\lambda}_{j};x,t)=0,\ \ \ \ j,k=1,2,\cdots,n.
	\end{equation}
	Then we should eliminate the zero structures of ${\bf P}^{\pm}(\lambda;x,t)$ at $\lambda_{k},\ \widehat{\lambda}_{j}$ to ensure that the RHP $1$ can be solved via the Plemelj formula. Combining the idea in \cite{yang2010nonlinear}, we assume
	\begin{equation}\label{3-G}
		{\bf G}^{+}(\lambda;x,t)=\mathbb{I}_{3}-{\bf v}{\bf N}^{-1}\widehat{\Gamma}\widehat{{\bf v}},\ \ \ \ \ ({\bf G}^{-}(\lambda;x,t))^{-1}=\mathbb{I}_{3}+{\bf v}\Gamma{\bf N}^{-1}\widehat{{\bf v}},
	\end{equation}
	where ${\bf N}(\lambda;x,t)$ is a matrix of order $n$ with its $(j,k)$-th element given by $n_{jk}=\frac{\langle\widehat{v}_{j}|v_{k}\rangle}{\lambda_{k}-\widehat{\lambda}_{j}},$ and
	\begin{equation*}
		\begin{split}
			&{\bf v}=\big(|v_{1}\rangle,\cdots,|v_{n}\rangle\big),\ \ \ \
			\Gamma={\rm diag}\big((\lambda-\lambda_{1})^{-1},\cdots,(\lambda-\lambda_{n})^{-1}\big),\\
			&\widehat{{\bf v}}=\big(\langle\widehat{v}_{1}|,\cdots,\langle\widehat{v}_{n}|\big)^{\top},\ \
			\widehat{\Gamma}={\rm diag}\big((\lambda-\widehat{\lambda}_{1})^{-1},\cdots,(\lambda-\widehat{\lambda}_{n})^{-1}\big).
		\end{split}
	\end{equation*}
	Then there is ${\bf G}^{+}(\lambda;x,t)({\bf G}^{-}(\lambda;x,t))^{-1} =\mathbb{I}_{3}.$ From the above, we can construct two matrix functions
	\begin{equation}
		\widetilde{{\bf M}}^{+}(\lambda;x,t)={\bf M}^{+}(\lambda;x,t)\big({\bf G}^{+}(\lambda;x,t)\big)^{-1},\ \ \ \ \  \widetilde{{\bf M}}^{-}(\lambda;x,t)={\bf M}^{-}(\lambda;x,t)\big({\bf G}^{-}(\lambda;x,t)\big)^{-1}.
	\end{equation}
	Then a new RHP can be constructed.\\
	{\bf Riemann-Hilbert problem $2$.\ }The matrix $\widetilde{{\bf M}}(\lambda;x,t)$ has the following properties:
	\begin{itemize}
		\item ${\bf Analyticity:}$\ $\widetilde{{\bf M}}^{\pm}(\lambda;x,t)$ is analytic in $\mathbb{C}^{\pm},$ respectively.
		\item ${\bf Jump\ condition:}$\ When $\lambda\in\mathbb{R}$, the jump relation between $\widetilde{{\bf M}}^{\pm}(\lambda;x,t)$ is as follows
		\begin{equation*}
			\widetilde{{\bf M}}^{+}(\lambda;x,t)=\widetilde{{\bf M}}^{-}(\lambda;x,t)\  \widetilde{{\bf J}}(\lambda;x,t),
		\end{equation*}
		where the jump matrix is
		\begin{equation*}
			\widetilde{{\bf J}}(\lambda;x,t)=	{\bf G}^{-}(\lambda;x,t){\bf J}(\lambda;x,t)\big({\bf G}^{+}(\lambda;x,t)\big)^{-1}.
		\end{equation*}
		\item ${\bf Normalization:}$\ $\widetilde{{\bf M }}(\lambda;x,t)\rightarrow\mathbb{I}_{3}$, as $\lambda\to \infty$.
	\end{itemize}
	Then the solution to the RHP $2$ can be obtained by applying the Plemelj formula
	\begin{equation}\label{3-sol-ir-re-RHP}
		\widetilde{\bf M}(\lambda)=\mathbb{I}_{3}+\frac{1}{2\pi \ii}\int_{\mathbb{R}}\frac{\widetilde{\bf M}^{-}(\xi)\left({\bf G}^{-}(\xi)({\bf J}(\xi)-\mathbb{I})({\bf G}^{+}(\xi))^{-1}\right)}{\xi-\lambda}d\xi.
	\end{equation}
	Thus, when $\lambda\rightarrow+\infty,$
	\begin{equation}\label{3-sol-T}
		{\bf M}^{+}(\lambda)= \mathbb{I}_{3}-\frac{1}{\lambda}\left(\frac{1}{2\pi \ii}\int_{\mathbb{R}}\widetilde{\bf M}^{-}(\xi)\Big({\bf G}^{-}(\xi)({\bf J}(\xi)-\mathbb{I})({\bf G}^{+}(\xi))^{-1}\Big)d\xi+{\bf v}{\bf N}^{-1}\widehat{{\bf v}}
		\right)+\mathcal{O}(\lambda^{-2}).
	\end{equation}
	To recover the potential functions $q_{1}(x,t)$ and $q_{2}(x,t)$, we expand ${\bf M}^{+}(\lambda;x,t)$ at large-$\lambda$ as
	\begin{equation}\label{3-M-expand}
		{\bf M}^{+}(\lambda;x,t)=\mathbb{I}_{3}+\frac{1}{\lambda}\ {\bf M}^{+}_{1}(x,t)+\frac{1}{\lambda^{2}}\ {\bf M}^{+}_{2}(x,t)+\cdots.
	\end{equation}
	Then the potential functions can be recovered by substituting \eqref{3-M-expand} 
	into \eqref{3-Psi} and comparing the coefficients of both sides of the equations to the same power of $\lambda$. From
	the coefficient of the highest power of $\lambda$, we can obtain
	\begin{equation}
		{\bf U}_{0}=\ii[\Lambda,{\bf M}^{+}_{1}]=\lim\limits_{\lambda\rightarrow+\infty}\ii\lambda[\Lambda,{\bf M}^{+}].
	\end{equation}
	Therefore
	\begin{equation}\label{3-sol-form-q}
		q_{1}(x,t)=2\ii\big({\bf M}^{+}_{1}(x,t)\big)_{12},\ \ \ \ q_{2}(x,t)=2\ii\big({\bf M}^{+}_{1}(x,t)\big)_{13}.
	\end{equation}

	\subsection{The explicit form of the fcHirota equation}\label{sec.2.3}
	Based on the idea in \cite{kaup1976closure,yang2010nonlinear}, we can construct the squared eigenfunctions $\{{\bf Z}_{1}^{\pm}(\lambda;x,t),\ {\bf Z}_{2}^{\pm}(\lambda;x,t)\}$ and the adjoint squared eigenfunctions $\{\Omega_{1}^{\pm}(\lambda;x,t),\ \Omega_{2}^{\pm}(\lambda;x,t)\}$, which are also the eigenfunctions of the recursion operator $\mathcal{L}$ and its adjoint $\mathcal{L}^{A}$, respectively,
	\begin{equation}
		\mathcal{L}{\bf Z}_{j}^{\pm}=2\lambda {\bf Z}_{j}^{\pm},\ \ \ \ \mathcal{L}^{A}\Omega_{j}^{\pm}=2\lambda \Omega_{j}^{\pm},\ \ j=1,2,
	\end{equation}
	where
	\begin{equation}\label{3-squared-eigenfunction}
		\begin{split}
			&{\bf Z}_{j}^{+}=\begin{bmatrix}
				\phi_{1,j+1}^{+}\ \widetilde{\phi}_{12}^{+}\\[4pt]
				\phi_{1,j+1}^{+}\ \widetilde{\phi}_{13}^{+}\\[4pt]
				\phi_{2,j+1}^{+}\ \widetilde{\phi}_{11}^{+}\\[4pt]
				\phi_{3,j+1}^{+}\ \widetilde{\phi}_{11}^{+}
			\end{bmatrix},\ \ 
			{\bf Z}_{j}^{-}=\begin{bmatrix}
				\phi_{11}^{+}\ \widetilde{\phi}_{j+1,2}^{+}\\[4pt]
				\phi_{11}^{+}\ \widetilde{\phi}_{j+1,3}^{+}\\[4pt]
				\phi_{21}^{+}\ \widetilde{\phi}_{j+1,1}^{+}\\[4pt]
				\phi_{31}^{+}\ \widetilde{\phi}_{j+1,1}^{+}
			\end{bmatrix},\ \ \ 
			\Omega_{j}^{+}=\begin{bmatrix}
				\ \ \ \widehat{\varphi}_{j1}\  \phi_{21}^{-}\\[4pt]
				\ \ \ \widehat{\varphi}_{j1}\  \phi_{31}^{-}\\[4pt]
				-\widehat{\varphi}_{j2}\  \phi_{11}^{-}\\[4pt]
				-\widehat{\varphi}_{j3}\ \phi_{11}^{-}
			\end{bmatrix},\ \ 
			\Omega_{j}^{-}=\begin{bmatrix}
				-\widetilde{\phi}^{-}_{11}\ \varphi_{2j}\\[4pt]
				-\widetilde{\phi}^{-}_{11}\  \varphi_{3j}\\[4pt]
				\ \ \ \widetilde{\phi}^{-}_{12}\ \varphi_{1j}\\[4pt]
				\ \ \ \widetilde{\phi}^{-}_{13}\ \varphi_{1j}
			\end{bmatrix},\ \ j=1,2,\\[4pt]
			&\begin{bmatrix}
				\varphi_{1}&\varphi_{2}
			\end{bmatrix}=
			\begin{bmatrix}
				\phi_{2}^{-}&\phi_{3}^{-}
			\end{bmatrix}
			\begin{bmatrix}
				s_{33}&-s_{23}\\[4pt]
				-s_{32}&s_{22}
			\end{bmatrix},\ \ \ \ \ \ \ 
			\begin{bmatrix}
				\widehat{\varphi}_{1}&\widehat{\varphi}_{2}
			\end{bmatrix}^{\top}=
			\begin{bmatrix}
				\widehat{s}_{33}&-\widehat{s}_{23}\\[4pt]
				-\widehat{s}_{32}&\widehat{s}_{22}
			\end{bmatrix}
			\begin{bmatrix}
				\widetilde{\phi}^{-}_{2}\\[4pt]
				\widetilde{\phi}^{-}_{3}
			\end{bmatrix},\\[4pt]
			&\ \ \widetilde{\Phi}^{\pm}=(\Phi^{\pm})^{-1}=\begin{bmatrix}
				\widetilde{\phi}^{\pm}_{1}& \widetilde{\phi}^{\pm}_{2}& \widetilde{\phi}^{\pm}_{3}
			\end{bmatrix}^{\top},\ \ \ \  \widetilde{\phi}^{\pm}_{k}=\begin{bmatrix}
				\widetilde{\phi}^{\pm}_{k1}&\widetilde{\phi}^{\pm}_{k2}&\widetilde{\phi}^{\pm}_{k3}
			\end{bmatrix},\ \ k=1,2,3,\\[4pt] &\ \ \ \varphi_{l}=\begin{bmatrix}
				\varphi_{1l}&\varphi_{2l}& \varphi_{3l}
			\end{bmatrix}^{\top},\ \ \ \  \widehat{\varphi}_{l}=\begin{bmatrix}
				\widehat{\varphi}_{l1}& \widehat{\varphi}_{l2}& \widehat{\varphi}_{l3}
			\end{bmatrix},\ \ l=1,2.
		\end{split}
	\end{equation}
	and the eigenfunctions ${\bf Z}_{j}^{\pm}(\lambda;x,t),\ \Omega_{j}^{\pm}(\lambda;x,t)\ (j=1,2)$ are all complete. Here, the superscript $^{\pm}$ of ${\bf Z}_{j}^{\pm}(\lambda;x,t)$ and $\Omega_{j}^{\pm}(\lambda;x,t)$ mean that they are analytic for $\lambda\in\mathbb{C}^{\pm}$, respectively. 
	Similarly,
	\begin{equation}
		\mathcal{F}_{h}(\mathcal{L}){\bf Z}_{j}^{\pm}=\mathcal{F}_{h}(2\lambda) {\bf Z}_{j}^{\pm},\ \ \ \ \mathcal{F}_{h}(\mathcal{L}^{A})\Omega_{j}^{\pm}=\mathcal{F}_{h}(2\lambda) \Omega_{j}^{\pm},\ \ j=1,2.
	\end{equation}
	
	Then we assume a sufficiently smooth and
	decaying vector function $\vartheta(x,t)=\begin{bmatrix}
		\vartheta_{1},&\vartheta_{2},&\vartheta_{3},&\vartheta_{4}
	\end{bmatrix}^{\top}$ which may be expanded in terms of the eigenfunctions  \eqref{3-squared-eigenfunction} as
	\begin{equation}
		\vartheta(x,t)=-\frac{1}{\pi}\bigg(\int_{\Gamma_{\mathbb{R}}^{+}}s_{11}^{-2}(\xi)\int_{\mathbb{R}}{\bf F}_{1}(\xi;x,y,t)\vartheta(y,t)\dd y\dd\xi+\int_{\Gamma_{\mathbb{R}}^{-}}\widehat{s}_{11}^{\ -2}(\xi)\int_{\mathbb{R}}{\bf F}_{2}(\xi;x,y,t)\vartheta(y,t)\dd y\dd\xi\bigg),
	\end{equation}
	where
	\begin{equation}
		{\bf F}_{1}(\lambda;x,y,t)=\sum\limits_{j=1}^{2}{\bf Z}_{j}^{+}(\lambda;x,t)\big(\Omega_{j}^{+}(\lambda;y,t)\big)^{\top},\ \ \ {\bf F}_{2}(\lambda;x,y,t)=\sum\limits_{j=1}^{2}{\bf Z}_{j}^{-}(\lambda;x,t)\big(\Omega_{j}^{-}(\lambda;y,t)\big)^{\top},
	\end{equation}
	and $\Gamma_{\mathbb{R}}^{\pm}$ are the semicircular contour in the upper and lower half plane evaluated
	from $-\infty$ to $+\infty$, respectively. Note that ${\bf F}_{1}(\lambda;x,y,t)$ is analytic in the upper half plane, while ${\bf F}_{2}(\lambda;x,y,t)$ is analytic in the lower half plane, and $s_{11}^{-2}(\lambda),\ \widehat{s}_{11}^{\ -2}(\lambda)$ are meromorphic in the upper and lower half-planes, respectively.
	
	So we can use $\mathcal{F}_{h}(\mathcal{L})$ to act on the vector function $\vartheta(x,t)$ to get
	\begin{multline}
		\mathcal{F}_{h}(\mathcal{L})\vartheta(x,t)=-\frac{1}{\pi}\bigg(\int_{\Gamma_{\mathbb{R}}^{+}}\mathcal{F}_{h}(2\lambda)s_{11}^{-2}(\xi)\int_{\mathbb{R}}{\bf F}_{1}(\xi;x,y,t)\vartheta(y,t)\dd y\dd\xi\\
		+\int_{\Gamma_{\mathbb{R}}^{-}}\mathcal{F}_{h}(2\lambda)\widehat{s}_{11}^{\ -2}(\xi)\int_{\mathbb{R}}{\bf F}_{2}(\xi;x,y,t)\vartheta(y,t)\dd y\dd\xi\bigg).
	\end{multline}
	Then we take $\vartheta(x,t)=\begin{bmatrix}
		\ \ \ {\bf q}^{\top}\\[3pt]
		-{\bf q}^{\dagger}
	\end{bmatrix}$, yields
	\begin{multline}\label{3-F(L)qr}
		\mathcal{F}_{h}(\mathcal{L})\begin{bmatrix}
			\ \ \ {\bf q}^{\top}\\[3pt]
			-{\bf q}^{\dagger}
		\end{bmatrix}=
		-\frac{1}{\pi}\bigg(\int_{\Gamma_{\mathbb{R}}^{+}}|4\xi^{2}|^{\epsilon}\ s_{11}^{-2}(\xi)\int_{\mathbb{R}}{\bf F}_{1}(\xi;x,y,t)\widehat{\vartheta}(y,t)\dd y\dd\xi\\
		+\int_{\Gamma_{\mathbb{R}}^{-}}|4\xi^{2}|^{\epsilon}\ \widehat{s}_{11}^{\ -2}(\xi)\int_{\mathbb{R}}{\bf F}_{2}(\xi;x,y,t)\widehat{\vartheta}(y,t)\dd y\dd\xi\bigg),
	\end{multline}
	where
	\begin{equation}\label{3-over-vartheta}
		\widehat{\vartheta}(x,t)=\begin{bmatrix}
			\widehat{\vartheta}_{1}\\[4pt]
			\widehat{\vartheta}_{2}\\[4pt]
			\widehat{\vartheta}_{3}\\[4pt]
			\widehat{\vartheta}_{4}
		\end{bmatrix}=
		\begin{bmatrix}
			-\frac{1}{2}\big(q_{1xx}+2q_{1}(|q_{1}|^{2}+|q_{2}|^{2})\big)-\ii\alpha\big(q_{1xxx}+3q_{1x}(2|q_{1}|^{2}+|q_{2}|^{2})+3q_{2x}q_{1}q_{2}^{*}\big)\\[4pt]
			-\frac{1}{2}\big(q_{2xx}+2q_{2}(|q_{1}|^{2}+|q_{2}|^{2})\big)-\ii\alpha\big(q_{2xxx}+3q_{1x}q_{1}^{*}q_{2}+3q_{2x}(|q_{1}|^{2}+2|q_{2}|^{2})\big)\\[4pt]
			\ \ \  \frac{1}{2}\big(q_{1xx}^{*}+2q_{1}^{*}(|q_{1}|^{2}+|q_{2}|^{2})\big)-\ii\alpha\big(q_{1xxx}^{*}+3q_{1x}^{*}(2|q_{1}|^{2}+|q_{2}|^{2})+3q_{2x}^{*}q_{1}^{*}q_{2}\big)\\[4pt]
			\ \ \ \frac{1}{2}\big(q_{2xx}^{*}+2q_{2}^{*}(|q_{1}|^{2}+|q_{2}|^{2})\big)-\ii\alpha\big(q_{2xxx}^{*}+3q_{1x}^{*}q_{1}q_{2}^{*}+3q_{2x}^{*}(|q_{1}|^{2}+2|q_{2}|^{2})\big)
		\end{bmatrix}.
	\end{equation}
	Therefore, we have the explicit representation of the fcHirota equation by combining the equations \eqref{3-three-scattering},\  \eqref{3-fcHirota-1} and \eqref{3-F(L)qr}, 
	\begin{equation*}
		\ii\widehat{{\bf q}}_{t}{+}\frac{1}{\pi}\bigg(\int_{\Gamma_{\mathbb{R}}^{+}}|4\xi^{2}|^{\epsilon}\ s_{11}^{-2}(\xi)\int_{\mathbb{R}}{\bf F}_{1}(\xi;x,y,t)\widehat{\vartheta}(y,t)\dd y\dd\xi {+}\int_{\Gamma_{\mathbb{R}}^{-}}|4\xi^{2}|^{\epsilon}\ \widehat{s}_{11}^{\ -2}(\xi)\int_{\mathbb{R}}{\bf F}_{2}(\xi;x,y,t)\widehat{\vartheta}(y,t)\dd y\dd\xi\bigg){=}0,
	\end{equation*}
	where $\widehat{{\bf q}}(x,t)=\begin{bmatrix}
		q_{1},&q_{2},&q_{1}^{*},&q_{2}^{*}
	\end{bmatrix}$. Then the fcHirota equation is given by
	\begin{equation}\label{3-fcHirota}
		\begin{cases}
			\ii q_{1t}+\frac{1}{\pi}\Big(\int_{\Gamma_{\mathbb{R}}^{+}}|4\lambda^{2}|^{\epsilon}\ s_{11}^{-2}(\lambda)\int_{\mathbb{R}}\tau_{11}(\lambda;x,y,t)\dd y\dd\lambda +\int_{\Gamma_{\mathbb{R}}^{-}}|4\lambda^{2}|^{\epsilon}\ \widehat{s}_{11}^{\ -2}(\lambda)\int_{\mathbb{R}}\tau_{12}(\lambda;x,y,t)\dd y\dd\lambda\Big)=0,\\[4pt]
			\ii q_{2t}+\frac{1}{\pi}\Big(\int_{\Gamma_{\mathbb{R}}^{+}}|4\lambda^{2}|^{\epsilon}\ s_{11}^{-2}(\lambda)\int_{\mathbb{R}}\tau_{21}(\lambda;x,y,t)\dd y\dd\lambda +\int_{\Gamma_{\mathbb{R}}^{-}}|4\lambda^{2}|^{\epsilon}\ \widehat{s}_{11}^{\ -2}(\lambda)\int_{\mathbb{R}}\tau_{22}(\lambda;x,y,t)\dd y\dd\lambda\Big)=0,
		\end{cases}
	\end{equation}
	where
	\begin{equation}\label{3-tau}
		\begin{split}
			\tau_{j1}(\lambda;x,y,t)=&-\widetilde{\phi}^{+}_{1,j+1}\phi_{11}^{-}\Big((\phi_{12}^{+}\widehat{\varphi}_{12}+\phi_{13}^{+}\widehat{\varphi}_{22})\widehat{\vartheta}_{3}+(\phi_{12}^{+}\widehat{\varphi}_{13}+\phi_{13}^{+}\widehat{\varphi}_{23})\widehat{\vartheta}_{4}\Big)\\[2pt]
			&+\widetilde{\phi}^{+}_{1,j+1}\big(\phi_{12}^{+}\widehat{\varphi}_{11}+\phi_{13}^{+}\widehat{\varphi}_{21}\big)\big(\phi_{21}^{-}\widehat{\vartheta}_{1}+\phi_{31}^{-}\widehat{\vartheta}_{2}\big),\\[4pt]
			\tau_{j2}(\lambda;x,y,t)=&-\phi_{11}^{+}\widetilde{\phi}_{11}^{-}\Big((\widetilde{\phi}^{+}_{2,j+1}\varphi_{21}+\widetilde{\phi}^{+}_{3,j+1}\varphi_{22})\widehat{\vartheta}_{1}+(\widetilde{\phi}^{+}_{2,j+1}\varphi_{31}+\widetilde{\phi}^{+}_{3,j+1}\varphi_{32})\widehat{\vartheta}_{2}\Big)\\[2pt]
			&+\phi_{11}^{+}\big(\widetilde{\phi}^{+}_{2,j+1}\varphi_{11}+\widetilde{\phi}^{+}_{3,j+1}\varphi_{12}\big)\big(\widetilde{\phi}_{12}^{-}\widehat{\vartheta}_{3}+\widetilde{\phi}_{13}^{-}\widehat{\vartheta}_{4}\big),\ \ \ \ \ \  j=1,2,
		\end{split}
	\end{equation}
	and in equation \eqref{3-tau}, $\widehat{\vartheta}_{l}=\widehat{\vartheta}_{l}(y,t)\ (l=1,\cdots,4)$,
	\begin{equation*}
		\begin{split}
			&\phi_{1k}^{+}=\phi_{1k}^{+}(\lambda;x,t),\ \ \ \  \widetilde{\phi}^{+}_{k,j+1}=\widetilde{\phi}_{k,j+1}^{+}(\lambda;x,t),\ \ \ \  \varphi_{kj}=\varphi_{kj}(\lambda;y,t),\\
			&\phi_{k1}^{-}=\phi_{k1}^{-}(\lambda;y,t),\ \ \ \  \widetilde{\phi}_{1k}^{-}=\widetilde{\phi}_{1k}^{-}(\lambda;y,t),\ \ \ \ \ \ \ \ \ \ \ \ \  \widehat{\varphi}_{jk}=\widehat{\varphi}_{jk}(\lambda;y,t),\ \ \ \ j=1,2,\ k=1,2,3.
		\end{split}
	\end{equation*}
	In particular, when $\epsilon=0$, the equation \eqref{3-fcHirota} will degenerate into the classical coupled Hirota equation.
	
	Based on the results obtained in this section, we will analyze the fractional soliton solutions of the fcHirota equation in the next section.

	\section{Fractional $n$-soliton solution}\label{3-sec.3}
	To find the fractional soliton solutions, we need to set $J(\lambda;x,t)=\mathbb{I}$ in the RHP $2$. This can be achieved by assuming $\widehat{s}_{12}(\lambda)=\widehat{s}_{13}(\lambda)=s_{21}(\lambda)=s_{31}(\lambda)=0,$ which means that there is no reflection in the scattering problem. Under these conditions,
	\begin{equation}\label{1-T1}
		{\bf M}^{+}(\lambda)= \mathbb{I}_{3}-\frac{1}{\lambda}\ {\bf v}{\bf N}^{-1}\widehat{{\bf v}}+\mathcal{O}(\lambda^{-2}),\ \ \lambda\in\mathbb{C^{+}}\rightarrow\infty.
	\end{equation}
	Below we can easily find the spatial and temporal
	evolutions for the vectors $|v_{k}\rangle$ and $\langle\widehat{v}_{j}|$. For example, let us take the $x$-derivative of both sides of the
	equation ${\bf P}^{+}(\lambda_{k};x,t)|v_{k}\rangle=0$. Then there is
	\begin{equation}
		|v_{k}(\lambda_{k};x)\rangle=\ee^{-\ii\lambda_{k}\Lambda x+\int^{x}_{x_{0}}\zeta_{k}(y)dy}\gamma_{k},\ \ \ \gamma_{k}=\begin{bmatrix}
			\gamma_{1k}&\gamma_{2k}&\gamma_{3k}
		\end{bmatrix}^{\top}.
	\end{equation}
	Without losing generality, we take $\zeta_{k}(x)=0,$ then
	\begin{equation}
		|v_{k}(\lambda_{k};x)\rangle=\ee^{-\ii\lambda_{k}\Lambda x}\gamma_{k}.
	\end{equation}
	Similarly, we can obtain another relation with respect to $t$. Summing up, we obtain
	\begin{equation}
		|v_{k}(\lambda_{k};x,t)\rangle=\ee^{-\ii\lambda_{k}\Lambda x+\ii\mathcal{F}_{h}(2\lambda_{k})\Omega t}\gamma_{k},\ \ 1\leq k\leq n,
	\end{equation}
	where $\gamma_{k}$ is independent of $x$ and $t$. And we can obtain the concrete representation of $\langle\widehat{v}_{j}|$
	\begin{equation}
		\langle\widehat{v}_{j}(\widehat{\lambda}_{j};x,t)|=\widehat{\gamma}_{j}\ee^{\ii\widehat{\lambda}_{j}\Lambda x-\ii\mathcal{F}_{h}(2\widehat{\lambda}_{j})\Omega t},\ \ \ \widehat{\gamma}_{j}=\begin{bmatrix}
			\widehat{\gamma}_{j1}&\widehat{\gamma}_{j2}&\widehat{\gamma}_{j3}
		\end{bmatrix},\ \ \ 1\leq j\leq n,
	\end{equation}
	where $\widehat{\gamma}_{j}$ is independent of $x$ and $t$. For convenience, we denote
	\begin{equation}\label{2-vk-hvk}
		|v_{k}(\lambda_{k};x,t)\rangle=
		\begin{bmatrix}
			\ee^{-\ii\lambda_{k}x}\gamma_{1k}\\[4pt]
			\ee^{\ii\theta_{\epsilon}(\lambda_{k};x,t)}\gamma_{2k}\\[4pt]
			\ee^{\ii\theta_{\epsilon}(\lambda_{k};x,t)}\gamma_{3k}
		\end{bmatrix},\ \ \ \ \ \ \ 
		\langle\widehat{v}_{j}(\widehat{\lambda}_{j};x,t)|=
		\begin{bmatrix}
			\ee^{\ii\widehat{\lambda}_{j}x}\widehat{\gamma}_{j1}\\[4pt]
			\ee^{-\ii\theta_{\epsilon}(\widehat{\lambda}_{j};x,t)}\widehat{\gamma}_{j2}\\[4pt]
			\ee^{-\ii\theta_{\epsilon}(\widehat{\lambda}_{j};x,t)}\widehat{\gamma}_{j3}
		\end{bmatrix}^{\top},
	\end{equation}
	where $\theta_{\epsilon}(\lambda;x,t)=\lambda x+2(1+4\alpha\lambda)\lambda^{2}|4\lambda^{2}|^{\epsilon}t$, 
	and $\gamma_{lk},\ \widehat{\gamma}_{jl},\ (l=1,2,3)$ are all complex constants. Furthermore, we know $s^{*}_{11}(\lambda^{*})=\widehat{s}_{11}(\lambda),\ \lambda\in\mathbb{C^{-}}$. Together with $s_{11}(\lambda_{k})=0,\ \widehat{s}_{11}(\widehat{\lambda}_{j})=0$, we can suppose
	\begin{equation}
		\widehat{\lambda}_{k}=\lambda_{k}^{*},\ \ \ \lambda_{k}\in\mathbb{C}^{+},\ \ k=1,\cdots,n.
	\end{equation}
	
	Based on the above preparations, we have the fractional $n$-soliton solution of the fcHirota equation as
	\begin{equation}\label{3-N-soliton}
		q_{1}^{[n]}(x,t)=-2\ii\frac{\det({\bf T}_{1})}{\det({\bf N})},\ \ \ \ \ \ \  q_{2}^{[n]}(x,t)=-2\ii\frac{\det({\bf T}_{2})}{\det({\bf N})},
	\end{equation}
	where
	\begin{equation}
		{\bf N}(\lambda;x,t)=
		\begin{bmatrix}
			\dfrac{\langle\widehat{v}_{1}|v_{1}\rangle}{\lambda_{1}-\widehat{\lambda}_{1}} &\cdots&\dfrac{\langle\widehat{v}_{1}|v_{n}\rangle}{\lambda_{n}-\widehat{\lambda}_{1}}\\[4pt]
			\vdots&\ddots&\vdots\\[4pt]
			\dfrac{\langle\widehat{v}_{n}|v_{1}\rangle}{\lambda_{1}-\widehat{\lambda}_{n}} &\cdots&\dfrac{\langle\widehat{v}_{n}|v_{n}\rangle}{\lambda_{n}-\widehat{\lambda}_{n}}
		\end{bmatrix},\ \ \ \ \ \ 
		{\bf T}_{j}(\lambda;x,t)=
		\begin{bmatrix}
			{\bf N}&-\widehat{v}^{(j+1)}\\[4pt]
			v^{(1)}&0
		\end{bmatrix},\ \ j=1,\ 2,
	\end{equation}
	here $v^{(j)},\ \widehat{v}^{(j)}$ denote the $j$-th row of ${\bf v}$ and $\widehat{{\bf v}}$, respectively, and
	\begin{equation}
		\begin{split}
			&	v^{(1)}=
			\begin{bmatrix}
				\ee^{-\ii\lambda_{1}x}\gamma_{11}&\ee^{-\ii\lambda_{2}x}\gamma_{12}&\cdots&\ee^{-\ii\lambda_{n}x}\gamma_{1n}
			\end{bmatrix},\\
			&\widehat{v}^{(j)}=
			\begin{bmatrix}
				\ee^{-\ii\theta_{\epsilon}(\widehat{\lambda}_{1};x,t)}\widehat{\gamma}_{1j}&\ee^{-\ii\theta_{\epsilon}(\widehat{\lambda}_{2};x,t)}\widehat{\gamma}_{2j}&\cdots&\ee^{-\ii\theta_{\epsilon}(\widehat{\lambda}_{n};x,t)}\widehat{\gamma}_{nj}
			\end{bmatrix}^{\top},\ \ j=1,\ 2,\ 3.
		\end{split}
	\end{equation}
	And the condition $\widehat{\gamma}_{11}\gamma_{jk}^{*}-\widehat{\gamma}_{kj}\gamma_{11}^{*}=0\ (j=2,3,\ k=1,\cdots,n)$, which can be derived from the symmetry property of ${\bf M}_{1}^{+}(\lambda;x,t)$, also need to be satisfied. 
	
		In what follows, we take $\widehat{\gamma}_{kj}=\gamma_{jk}^{*}$ ($j=1,2,3,k=1,\cdots,n$), then $\langle\widehat{v}_{k}(\widehat{\lambda}_{k};x,t)|=|v_{k}(\lambda_{k};x,t)\rangle^{\dagger}\ (k=1,\cdots,n)$. And we will discuss the fractional one- and two-soliton solutions of the fcHirota equation in terms of the equation \eqref{3-N-soliton} with $n=1,2$.
	
	\subsection{Fractional one-soliton solution}\label{subsection.one-soliton}
	We choose the spectral parameter $\lambda_{1}=a_{1}+\ii b_{1}$ with $b_{1}\neq 0$, then the expression of fractional one-soliton solution of the fcHirota equation can be written as
	\begin{equation}\label{3-sol-1}
		q_{1}^{[1]}(x,t)=2b_{1}\gamma_{21}^{*}\big(\gamma_{11}^{*}\sqrt{g_{1}}\big)^{-1}\ee^{f_{2}}{\rm sech}(f_{1}),\ \ \ \ q_{2}^{[1]}(x,t)=2b_{1}\gamma_{31}^{*}\big(\gamma_{11}^{*}\sqrt{g_{1}}\big)^{-1}\ee^{f_{2}}{\rm sech}(f_{1}),
	\end{equation}
	where
	\begin{equation}\label{3-f-g1}
		\begin{split}
			&f_{1}=2b_{1}x+4^{1+\epsilon}b_{1}\big(a_{1}+2\alpha(3a_{1}^{2}-b_{1}^{2})\big)(a_{1}^{2}+b_{1}^{2})^{\epsilon}t-\ln\sqrt{g_{1}},\\[2pt]
			&f_{2}=-2\ii a_{1}x-2^{1+2\epsilon}\ii\big(4\alpha a_{1}(a_{1}^{2}-3b_{1}^{2})+(a_{1}^{2}-b_{1}^{2})\big)(a_{1}^{2}+b_{1}^{2})^{\epsilon}t,\\[2pt]
			& g_{1}=\dfrac{|\gamma_{21}|^{2}+|\gamma_{31}|^{2}}{|\gamma_{11}|^{2}}.
		\end{split}
	\end{equation}
	The graphs of this soliton are illustrated in Fig.\ref{fig:sol-1} by choosing the different values of $\epsilon$.
	\begin{figure}
		\centering
		\includegraphics[width=0.9\linewidth]{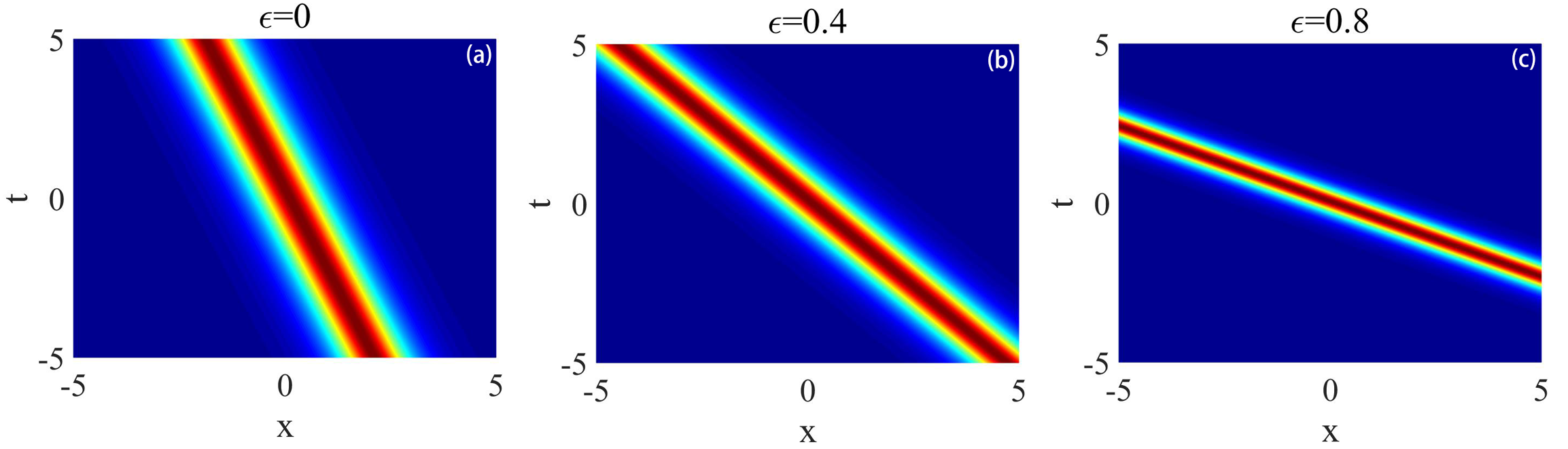}
		\caption{\small{Fractional one-soliton solution \eqref{3-sol-1}. Choosing the parameters: $\gamma_{11}=1,\ \gamma_{21}=1,\ \gamma_{31}=2,\  a_{1}=1,\ b_{1}=1,\ \alpha=-\frac{1}{5}.$}}
		\label{fig:sol-1}
	\end{figure}
	One may notice that the solutions of $q_{1}^{[1]}(x,t)$ and $q_{2}^{[1]}(x,t)$ are basically the same, while their
	amplitudes are different, which is reflected in the difference of the values of $\gamma_{21}^{*}$ and $\gamma_{31}^{*}$, and it can be observed in Fig.\ref{fig:q-1-tv}(a)-(b).
	\begin{figure}
		\centering
		\includegraphics[width=0.9\linewidth]{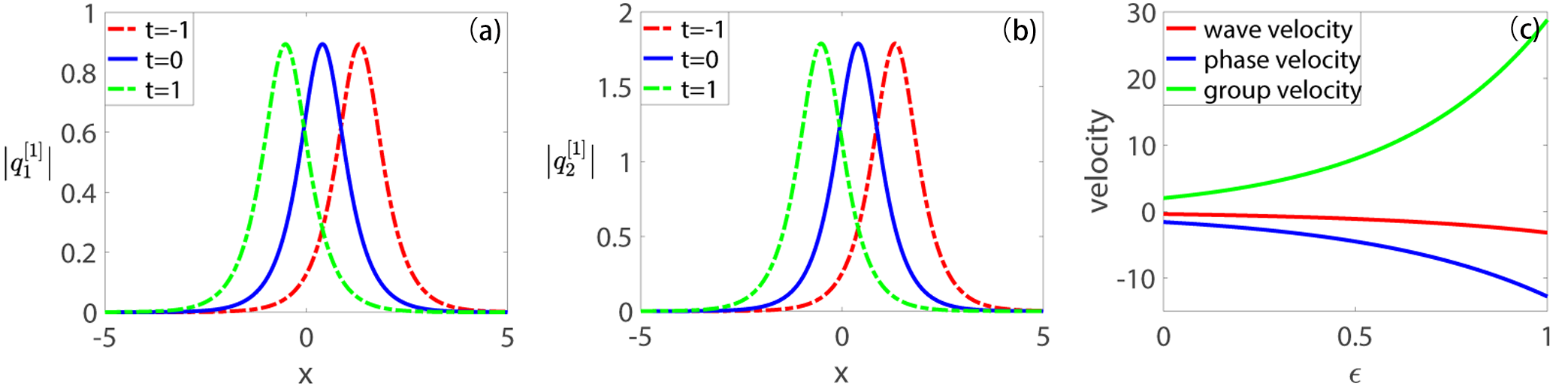}
		\caption{\small{(a) and (b): the direction of wave propagation. The parameter values are the same as in Fig.\ref{fig:sol-1},} and $\epsilon=\frac{2}{5}$. (c): wave velocity, phase velocity and group velocity of the solution \eqref{3-sol-1}.}
		\label{fig:q-1-tv}
	\end{figure}
	
	Through the analysis, we find that the wave velocity $v_{w}(a_{1},b_{1},\alpha,\epsilon)$, phase velocity $v_{p}(a_{1},b_{1},\alpha,\epsilon)$, and group velocity $v_{g}(a_{1},b_{1},\alpha,\epsilon)$ of the fractional one-soliton
	solutions $q_{1}^{[1]}(x,t)$ and $q_{2}^{[1]}(x,t)$ are the same. Furthermore, we can give their forms from the solution \eqref{3-sol-1}
	\begin{equation}\label{3-velocity}
		\begin{split}
			&v_{w}(a_{1},b_{1},\alpha,\epsilon)=-2^{1+2\epsilon}\big(a_{1}+2\alpha(3a_{1}^{2}-b_{1}^{2})\big)(a_{1}^{2}+b_{1}^{2})^{\epsilon},\\
			&v_{p}(a_{1},b_{1},\alpha,\epsilon)=-4^{\epsilon}\big(4\alpha(a_{1}^{2}-3b_{1}^{2})+a_{1}^{-1}(a_{1}^{2}-b_{1}^{2})\big)(a_{1}^{2}+b_{1}^{2})^{\epsilon},\\
			&v_{g}(a_{1},b_{1},\alpha,\epsilon)=-2^{1+2\epsilon}\left(\epsilon a_{1}\big((a_{1}^{2}-b_{1}^{2})+4\alpha a_{1}(a_{1}^{2}-3b_{1}^{2})\big)+(a_{1}^{2}+b_{1}^{2})\big(a_{1}+6\alpha(a_{1}^{2}-b_{1}^{2})\big)\right)(a_{1}^{2}+b_{1}^{2})^{\epsilon-1}.
		\end{split}
	\end{equation}
	According to the parameter values in Fig.\ref{fig:sol-1} and the equation \eqref{3-velocity}, $v_{w}=-0.4\times 8^{\epsilon}<0$. So the soliton is left-going travelling-wave solitons (see Fig.\ref{fig:q-1-tv}(a)-(b)) at $\epsilon=0,\ 0.4,\ 0.8$. Furthermore, with the increase of $\epsilon$, the velocity of the travelling-wave will be faster (see Fig.\ref{fig:q-1-tv}(c)).

	\subsection{Fractional two-soliton solution}
	For $n=2$, we choose the spectral parameters $\lambda_{1}=a_{1}+\ii b_{1},\ \lambda_{2}=a_{2}+\ii b_{2}$, and $\widehat{\lambda}_{j}=\lambda_{j}^{*}$ ($j=1,2$), then the expression of the fractional two-soliton solution of the fcHirota equation can be expressed as
	\begin{equation}\label{3-sol-2}
		q_{1}^{[2]}(x,t)=-2\ii\frac{\det({\bf T}_{1})}{\det({\bf N})},\ \ \ \ \ \ \  q_{2}^{[2]}(x,t)=-2\ii\frac{\det({\bf T}_{2})}{\det({\bf N})},
	\end{equation}
	where
	\begin{equation}
		\begin{split}
			&{\bf N}(\lambda;x,t)=
			\begin{bmatrix}
				\dfrac{\langle\widehat{v}_{1}|v_{1}\rangle}{\lambda_{1}-\widehat{\lambda}_{1}} &\dfrac{\langle\widehat{v}_{1}|v_{2}\rangle}{\lambda_{2}-\widehat{\lambda}_{1}}\\[15pt]
				\dfrac{\langle\widehat{v}_{2}|v_{1}\rangle}{\lambda_{1}-\widehat{\lambda}_{2}} &\dfrac{\langle\widehat{v}_{2}|v_{2}\rangle}{\lambda_{2}-\widehat{\lambda}_{2}}
			\end{bmatrix},\ \ \ 
			{\bf T}_{j-1}(\lambda;x,t)=
			\begin{bmatrix}
				\dfrac{\langle\widehat{v}_{1}|v_{1}\rangle}{\lambda_{1}-\widehat{\lambda}_{1}} &\dfrac{\langle\widehat{v}_{1}|v_{2}\rangle}{\lambda_{2}-\widehat{\lambda}_{1}} &-\ee^{-\ii\theta_{\epsilon}(\widehat{\lambda}_{1};x,t)}\widehat{\gamma}_{1j}\\[12pt]
				\dfrac{\langle\widehat{v}_{2}|v_{1}\rangle}{\lambda_{1}-\widehat{\lambda}_{2}} &\dfrac{\langle\widehat{v}_{2}|v_{2}\rangle}{\lambda_{2}-\widehat{\lambda}_{2}} &-\ee^{-\ii\theta_{\epsilon}(\widehat{\lambda}_{2};x,t)}\widehat{\gamma}_{2j}\\[12pt]
				\ee^{-\ii\lambda_{1}x}\gamma_{11}&\ee^{-\ii\lambda_{2}x}\gamma_{12}&0
			\end{bmatrix},
		\end{split}
	\end{equation}
and $\widehat{\gamma}_{kj}=\gamma_{jk}^{*}\ (j=1,2,3,k=1,2)$.
	The graphs of the solitons are illustrated in Fig.\ref{fig:sol-2}. From Figs.\ref{fig:sol-2}-\ref{fig:q-2-t}, we know that the solitons are right-going travelling-wave solitons at $\epsilon=0,\ 0.4,\ 0.8$. Moreover, with the increase of $\epsilon$, the velocity of the travelling-wave will be slower.
	\begin{figure}
		\centering
		\includegraphics[width=0.9\linewidth]{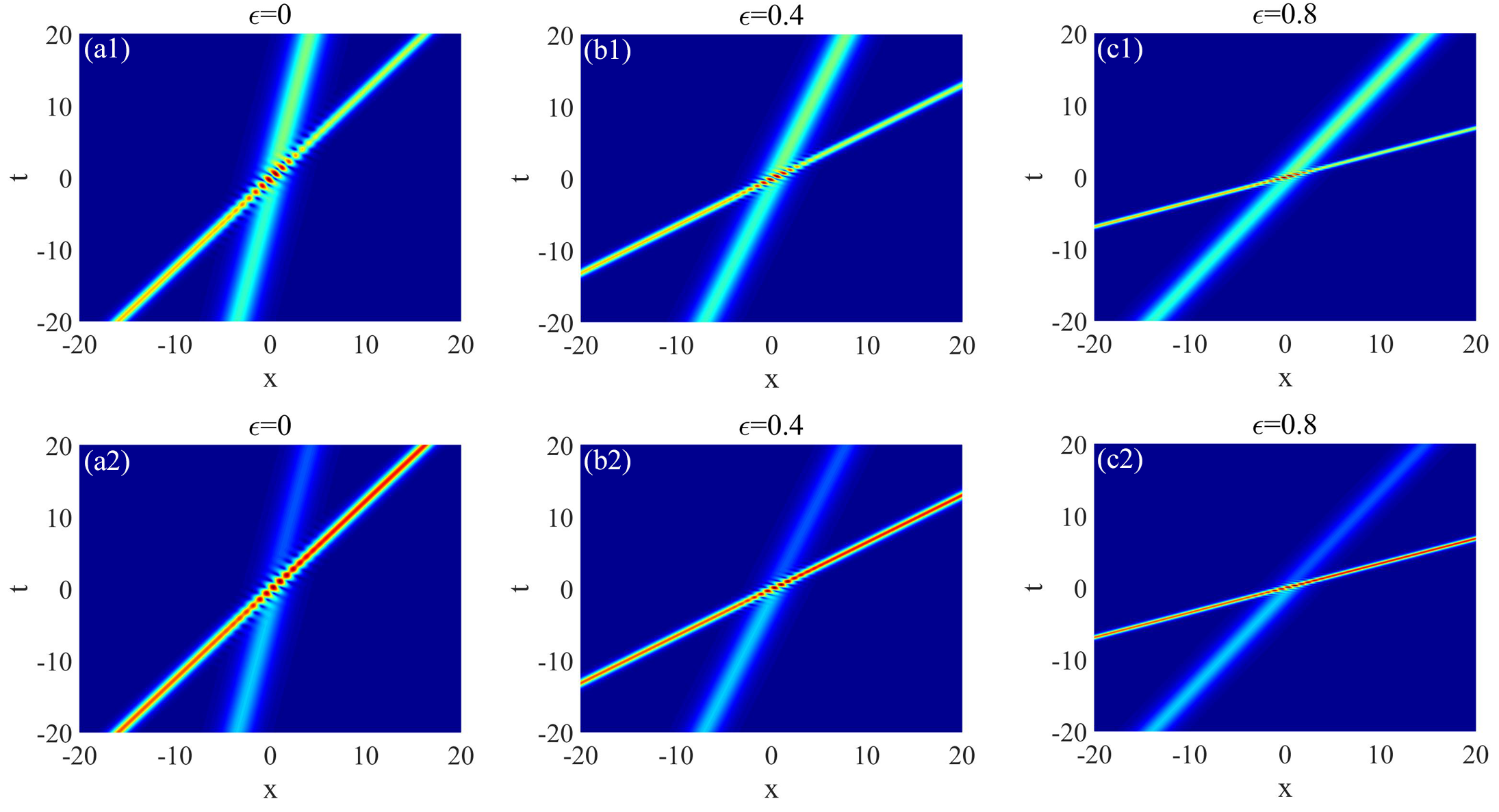}
		\caption{\small{Fractional two-soliton solution \eqref{3-sol-2}. Upper row: the graphs of the solution $\big|q_{1}^{[2]}\big|$. lower row: the graphs of the solution $\big|q_{2}^{[2]}\big|$. The parameters are as follows: $\gamma_{11}=2+\ii,\ \gamma_{21}=1-\ii,\ \gamma_{31}=-2+\ii,\ \gamma_{12}=1+\ii,\ \gamma_{22}=-1+\ii,\ \gamma_{32}=1+\frac{1}{2}\ii,\ a_{1}=-\frac{1}{2},\ b_{1}=1,\ a_{2}=1,\ b_{2}=\frac{1}{2},\  \alpha=-\frac{1}{5}.$}}
		\label{fig:sol-2}
	\end{figure}
	\begin{figure}
		\centering
		\includegraphics[width=0.8\linewidth]{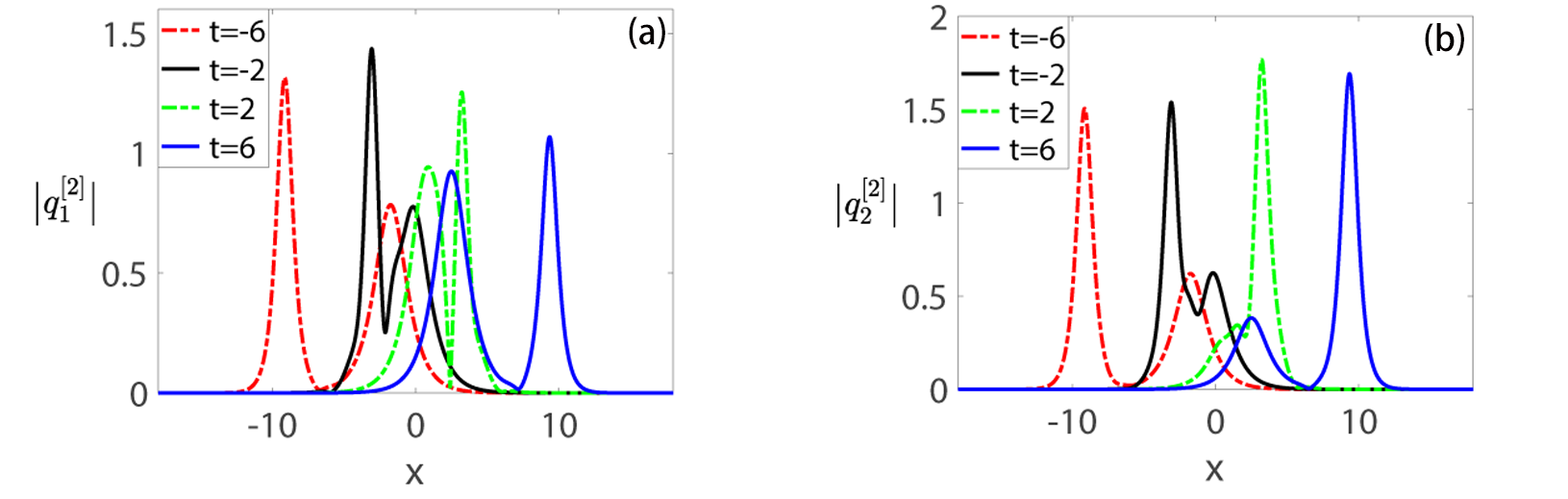}
		\caption{\small{The direction of wave propagation. The parameter values are the same as in Fig.\ref{fig:sol-2}, and $\epsilon=\frac{2}{5}$.}}
		\label{fig:q-2-t}
	\end{figure}
	
	Next, we will analyze the asymptotic analysis of the fractional two-soliton solution.
	\begin{theorem}\label{Th-asy}
		For the fractional coupled Hirota equation, we define $\widehat{\lambda}_{j}=\lambda_{j}^{*},\  \lambda_{j}\in\mathbb{C}\backslash\mathbb{R}\ (j=1,2)$, then its fractional two-soliton solutions $q_{1}^{[2]}(x,t)$ and $q_{2}^{[2]}(x,t)$ can be approximated as the sum of two fractional one-soliton solutions, as $|t|\to \infty$.
	\end{theorem}
	\begin{proof}
		We suppose $\lambda_{1}=a_{1}+\ii b_{1},\ \lambda_{2}=a_{2}+\ii b_{2}$. 
		Firstly, we know
		\begin{equation*}
			|v_{k}(\lambda_{k};x,t)\rangle=
			\begin{bmatrix}
				\ee^{-\ii\lambda_{k}x}\gamma_{1k}&
				\ee^{\ii\theta_{\epsilon}(\lambda_{k};x,t)}\gamma_{2k}&
				\ee^{\ii\theta_{\epsilon}(\lambda_{k};x,t)}\gamma_{3k}
			\end{bmatrix}^{\top},
		\end{equation*}
		and ${\rm Re}\big(\ii\theta_{\epsilon}(\lambda_{k};x,t)\big)=-b_{k}(x+B_{k}t)$, where
		\begin{equation*}
			B_{k}=4^{1+\epsilon}\big(2\alpha(3a_{k}^{2}-b_{k}^{2})+a_{k}\big)(a_{k}^{2}+b_{k}^{2})^{\epsilon},\ \ k=1,\ 2.
		\end{equation*}
		Now we fix ${\rm Re}\big(\ii\theta_{\epsilon}(\lambda_{1};x,t)\big)=c_{1}\ (|c_{1}|\ll\infty),$ then
		\begin{equation*}
			{\rm Re}\big(\ii\theta_{\epsilon}(\lambda_{2};x,t)\big)
			=-b_{2}(x+B_{2}t)
			=b_{2}\left((B_{1}-B_{2})t+\frac{c_{1}}{b_{1}}\right).
		\end{equation*}
		Without loss of generality, we assume that $b_{1},\ b_{2},\ B_{2}-B_{1}>0$.
		It is easy to find that when $t\rightarrow+\infty$, ${\rm Re}\big(\ii\theta_{\epsilon}(\lambda_{2};x,t)\big)\rightarrow-\infty$, i.e., $|v_{2}\rangle\rightarrow
		\begin{bmatrix}
			1,&0,&0
		\end{bmatrix}^{\top}$. 
		Then we have
		\begin{equation*}
			\begin{split}
				{\bf N}\rightarrow&
				\begin{bmatrix}
					\dfrac{|\gamma_{11}|^{2}\ee^{-\ii(\lambda_{1}-\lambda_{1}^{*})x}+(|\gamma_{21}|^{2}+|\gamma_{31}|^{2})\ee^{\ii\theta_{\epsilon}(\lambda_{1};x,t)-\ii\theta_{\epsilon}^{*}(\lambda_{1};x,t)}
					}{\lambda_{1}-\lambda_{1}^{*}}&\dfrac{\gamma_{11}^{*}\ee^{\ii\lambda_{1}^{*}x}}{\lambda_{2}-\lambda_{1}^{*}}\\[10pt]
					\dfrac{\gamma_{11}\ee^{-\ii\lambda_{1}x}}{\lambda_{1}-\lambda_{2}^{*}}&\dfrac{1}{\lambda_{2}-\lambda_{2}^{*}}\\
				\end{bmatrix}, \\[8pt]
				{\bf T}_{1}\rightarrow&\left[
				\begin{array}{cccc}
					\multicolumn{2}{c}{\raisebox{-1.5ex}[0pt]{${\bf N}$}}
					& -\gamma_{21}^{*}\ee^{-\ii\theta_{\epsilon}^{*}(\lambda_{1};x,t)}\\[4pt]
					& & 0\\[4pt]
					\gamma_{11}\ee^{-\ii\lambda_{1}x} & 1 & 0
				\end{array}\right],\ \ \ \ {\rm as}\ \ t\to+\infty.
			\end{split}
		\end{equation*}
		Based on the solution \eqref{3-sol-2},
		\begin{equation*}
			q^{[2]}_{11,+}(x,t)=2b_{1}\gamma_{21}^{*}(\gamma_{11}^{*}\sqrt{g_{1}})^{-1} \ee^{h_{1}-\ii\eta_{1i}}\ {\rm sech}(h_{2}-\eta_{1r})+\mathcal{O}\big(\ee^{b_{2}(B_{1}-B_{2})t}\big),
		\end{equation*}
		where the subscript $_{+}$ refers to the case of $t\to +\infty$, $g_{1}$ is given by \eqref{3-f-g1},
		\begin{equation*}
			\begin{split}
				&h_{1}=-\ii\Big(2a_{1}x+2^{1+2\epsilon}\big(4\alpha a_{1}(a_{1}^{2}-3b_{1}^{2})+a_{1}^{2}-b_{1}^{2}\big)(a_{1}^{2}+b_{1}^{2})^{\epsilon}t\Big),\\
				&h_{2}=2b_{1}x+4^{1+\epsilon}\big(2\alpha b_{1}(3a_{1}^{2}-b_{1}^{2})+a_{1}b_{1}\big)(a_{1}^{2}+b_{1}^{2})^{\epsilon}t,\\
				&\eta_{1}=\ln\Big(\sqrt{g_{1}}\ \frac{\lambda_{1}-\lambda_{2}^{*}}{\lambda_{1}-\lambda_{2}}\Big):=\eta_{1r}+\ii\eta_{1i}.
			\end{split}
		\end{equation*}
		When $t{\rightarrow}-\infty$, ${\rm Re}\big(\ii\theta_{\epsilon}(\lambda_{2};x,t)\big){\rightarrow}+\infty$, that is, $|v_{2}\rangle{\to}\begin{bmatrix}
			0,&\gamma_{22},&\gamma_{32}
		\end{bmatrix}^{\top}$. 
		Then 
		\begin{equation*}
			\begin{split}
				{\bf N}\rightarrow&
				\begin{bmatrix}
					\dfrac{|\gamma_{11}|^{2}\ee^{-\ii(\lambda_{1}-\lambda_{1}^{*})x}+(|\gamma_{21}|^{2}+|\gamma_{31}|^{2})\ee^{\ii\theta_{\epsilon}(\lambda_{1};x,t)-\ii\theta_{\epsilon}^{*}(\lambda_{1};x,t)}}{\lambda_{1}-\lambda_{1}^{*}}&\dfrac{(\gamma_{21}^{*}\gamma_{22}+\gamma_{31}^{*}\gamma_{32})\ee^{-\ii\theta_{\epsilon}^{*}(\lambda_{1};x,t)x}}{\lambda_{2}-\lambda_{1}^{*}}\\[10pt]
					\dfrac{(\gamma_{21}\gamma_{22}^{*}+\gamma_{31}\gamma_{32}^{*})\ee^{\ii\theta_{\epsilon}(\lambda_{1};x,t)x}}{\lambda_{1}-\lambda_{2}^{*}}&\dfrac{|\gamma_{22}|^{2}+|\gamma_{32}|^{2}}{\lambda_{2}-\lambda_{2}^{*}}\\
				\end{bmatrix}, \\[8pt]
				{\bf T}_{1}\rightarrow&\left[
				\begin{array}{cccc}
					\multicolumn{2}{c}{\raisebox{-1.5ex}[0pt]{${\bf N}$}}
					& -\gamma_{21}^{*}\ee^{-\ii\theta_{\epsilon}^{*}(\lambda_{1};x,t)}\\[4pt]
					& & -\gamma_{22}^{*}\\[4pt]
					\gamma_{11}\ee^{-\ii\lambda_{1}x} & 0 & 0
				\end{array}\right],\ \ \ \ {\rm as}\ \ t\to-\infty.
			\end{split}
		\end{equation*}
		Therefore,
		\begin{equation*}
			\begin{split}
				q_{11,-}^{[2]}(x,t)=&\frac{2b_{1}}{\sqrt{\eta_{2}}}\bigg(\frac{\gamma_{21}^{*}}{\gamma_{11}^{*}}+\frac{2\ii b_{2}\gamma_{22}^{*}(\gamma_{21}^{*}\gamma_{22}+\gamma_{31}^{*}\gamma_{32})}{\gamma_{11}^{*}(|\gamma_{22}|^{2}+|\gamma_{32}|^{2})\big(a_{1}-a_{2}-\ii(b_{1}+b_{2})\big)}\bigg)\ee^{h_{1}}\ {\rm sech}(h_{2}-\ln\sqrt{\eta_{2}})\\
				&+\mathcal{O}\big(\ee^{-b_{2}(B_{1}-B_{2})t}\big),
			\end{split}
		\end{equation*}
		where the subscript $_{-}$ refers to the case of $t\to -\infty$,
		\begin{equation*}
			\eta_{2}=g_{1}-\frac{4b_{1}b_{2}g_{2}}{|\lambda_{1}-\lambda_{2}^{*}|^{2}},\ \ \ g_{2}=\dfrac{|\gamma_{21}\gamma_{22}^{*}+\gamma_{31}\gamma_{32}^{*}|^{2}}{|\gamma_{11}|^{2}\big(|\gamma_{22}|^{2}+|\gamma_{32}|^{2}\big)}.
		\end{equation*}
		Note that  $\eta_{2}>0$ can be obtained through analysis.
		
		Similarly, we can analyze $|v_{1}(\lambda_{k};x,t)\rangle$ by fixing ${\rm Re}\big(\ii\theta_{\epsilon}(\lambda_{2};x,t)\big)=c_{2}\ (|c_{2}|\ll\infty)$, then we can obtain
		\begin{equation*}
			\begin{split}
				q_{12,+}^{[2]}(x,t)=&\frac{2b_{2}}{\sqrt{\eta_{3}}}\bigg(\frac{\gamma_{22}^{*}}{\gamma_{12}^{*}}-\frac{2\ii b_{1}\gamma_{21}^{*}(\gamma_{21}\gamma_{22}^{*}+\gamma_{31}\gamma_{32}^{*})}{\gamma_{12}^{*}(|\gamma_{21}|^{2}+|\gamma_{31}|^{2})\big(a_{1}-a_{2}+\ii(b_{1}+b_{2})\big)}\bigg)\ee^{h_{3}}\ {\rm sech}(h_{4}-\ln\sqrt{\eta_{3}})\\
				&+\mathcal{O}\big(\ee^{b_{1}(B_{1}-B_{2})t}\big),\\
				q_{12,-}^{[2]}(x,t)=&2b_{2}\gamma_{22}^{*}(\gamma_{12}^{*}\sqrt{g_{3}})^{-1} \ee^{h_{3}-\ii\eta_{4i}}\ {\rm sech}(h_{4}-\eta_{4r})+\mathcal{O}\big(\ee^{-b_{1}(B_{1}-B_{2})t}\big),
			\end{split}
		\end{equation*}
		and
		\begin{equation*}
			\begin{split}
				&h_{3}=-\ii\Big(2a_{2}x+2^{1+2\epsilon}\big(4\alpha a_{2}(a_{2}^{2}-3b_{2}^{2})+a_{2}^{2}-b_{2}^{2}\big)(a_{2}^{2}+b_{2}^{2})^{\epsilon}t\Big),\\
				&h_{4}=2b_{2}x+4^{1+\epsilon}\big(2\alpha b_{2}(3a_{2}^{2}-b_{2}^{2})+a_{2}b_{2}\big)(a_{2}^{2}+b_{2}^{2})^{\epsilon}t,\\
				& \eta_{3}=g_{3}-\frac{4b_{1}b_{2}g_{4}}{|\lambda_{1}-\lambda_{2}^{*}|^{2}},\ \ \ \ \ \eta_{4}=\ln\Big(\sqrt{g_{3}}\ \frac{\lambda_{1}^{*}-\lambda_{2}}{\lambda_{1}-\lambda_{2}}\Big):=\eta_{4r}+\ii\eta_{4i},\\
				&g_{3}=\frac{|\gamma_{22}|^{2}+|\gamma_{32}|^{2}}{|\gamma_{12}|^{2}},\ \ \ \ \ \ \ \ \ \ \ \ \ \ \ \  g_{4}=\dfrac{|\gamma_{21}\gamma_{22}^{*}+\gamma_{31}^{*}\gamma_{32}|^{2}}{|\gamma_{12}|^{2}\big(|\gamma_{21}|^{2}+|\gamma_{31}|^{2}\big)}.
			\end{split}
		\end{equation*}
		Note that  $\eta_{3}>0$ can also be obtained through analysis. Then
		\begin{multline*}
			q_{1,+}^{[2]}(x,t)=
			\frac{2b_{2}}{\sqrt{\eta_{3}}}\Big(\frac{\gamma_{22}^{*}}{\gamma_{12}^{*}}-\frac{2\ii b_{1}\gamma_{21}^{*}(\gamma_{21}\gamma_{22}^{*}+\gamma_{31}\gamma_{32}^{*})}{\gamma_{12}^{*}(|\gamma_{21}|^{2}+|\gamma_{31}|^{2})\big(a_{1}-a_{2}+\ii(b_{1}+b_{2})\big)}\Big)\ee^{h_{3}}\ {\rm sech}(h_{4}-\ln\sqrt{\eta_{3}})\\[2pt]
			+2b_{1}\gamma_{21}^{*}(\gamma_{11}^{*}\sqrt{g_{1}})^{-1} \ee^{h_{1}-\ii\eta_{1i}}\ {\rm sech}(h_{2}-\eta_{1r})+\mathcal{O}\big(\ee^{\iota_{1}}\big),
		\end{multline*}
		\begin{multline*}
			q_{1,-}^{[2]}(x,t)=\frac{2b_{1}}{\sqrt{\eta_{2}}}\Big(\frac{\gamma_{21}^{*}}{\gamma_{11}^{*}}+\frac{2\ii b_{2}\gamma_{22}^{*}(\gamma_{21}^{*}\gamma_{22}+\gamma_{31}^{*}\gamma_{32})}{\gamma_{11}^{*}(|\gamma_{22}|^{2}+|\gamma_{32}|^{2})\big(a_{1}-a_{2}-\ii(b_{1}+b_{2})\big)}\Big)\ee^{h_{1}}\ {\rm sech}(h_{2}-\ln\sqrt{\eta_{2}})\\[2pt]
			+2b_{2}\gamma_{22}^{*}(\gamma_{21}^{*}\sqrt{g_{3}})^{-1} \ee^{h_{3}-\ii\eta_{4i}}\ {\rm sech}(h_{4}-\eta_{4r})+\mathcal{O}\big(\ee^{\iota_{2}}\big),
		\end{multline*}
		where
		\begin{equation*}
			\iota_{1}=\min\big(b_{1}(B_{1}-B_{2})t,b_{2}(B_{1}-B_{2})t\big),\ \ \ \iota_{2}=\min\big(-b_{1}(B_{1}-B_{2})t,-b_{2}(B_{1}-B_{2})t\big).
		\end{equation*}
		The asymptotic expression of $q_{2}^{[2]}(x,t)$ can also be obtained by using the same method as above.
	\end{proof}
	We take the solution $q_{1}^{[2]}(x,t)$ as an example to compare the exact solution $q_{1}^{[2]}(x,t)$ and the asymptotic solutions $q_{1,\pm}^{[2]}(x,t)$ at $t=\pm 10$, respectively. The relevant results can be verified in Fig.\ref{fig:asy}.
	\begin{figure}
		\centering
		\includegraphics[width=0.9\linewidth]{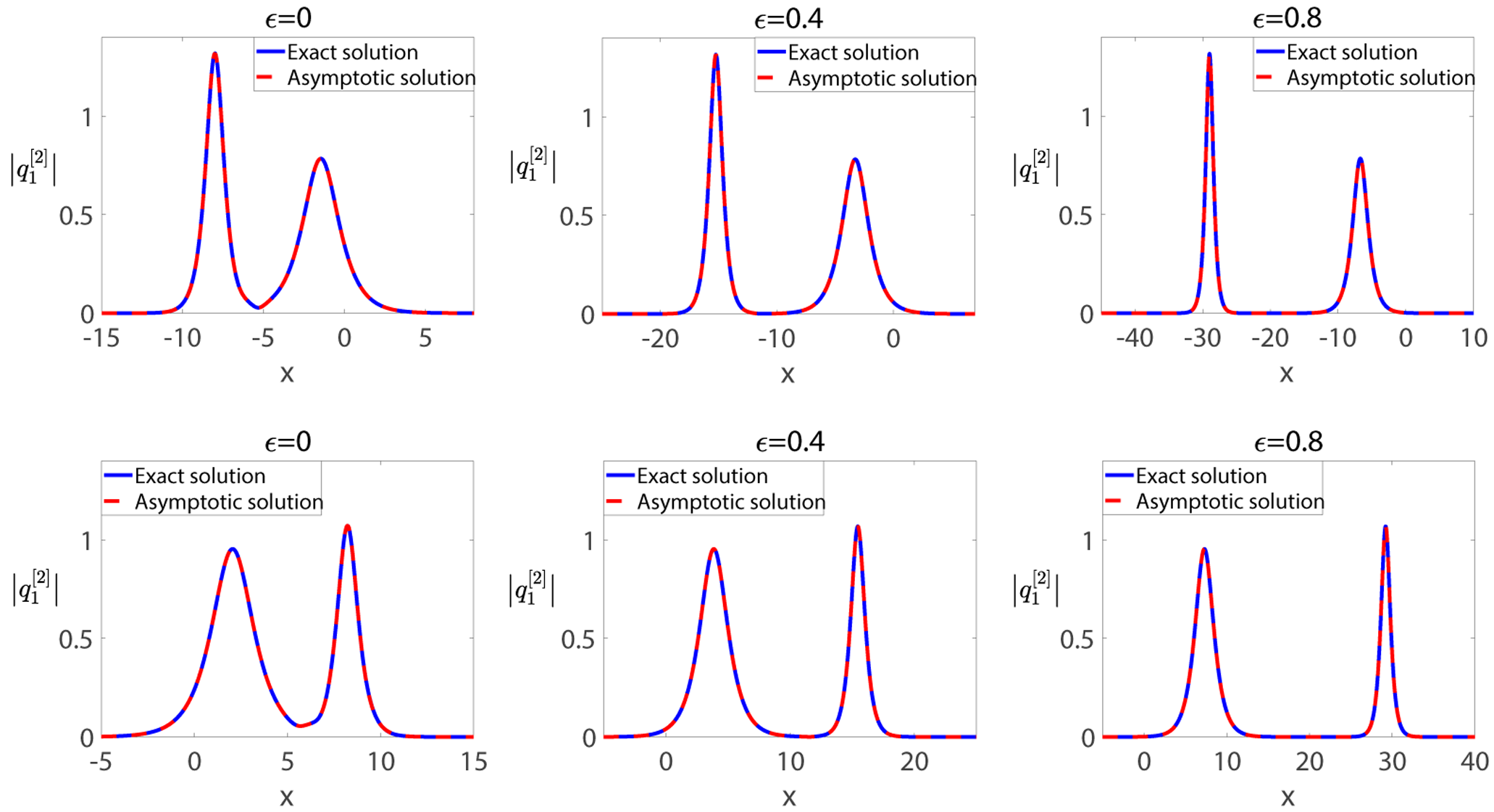}
		\caption{\small{The comparision between the exact solutions and the asymptotic solutions. The upper row is at $t=-10$, and the lower row is at $t=10$.}}
		\label{fig:asy}
	\end{figure}

	Based on the proof process of the theorem \ref{Th-asy}, we assume that the amplitudes of $\big|q_{11,-}^{[2]}\big|$,   $\big|q_{12,-}^{[2]}\big|$, $\big|q_{11,+}^{[2]}\big|$, and $\big|q_{12,+}^{[2]}\big|$ are $A_{11,-}$, $A_{12,-}$, $A_{11,+}$, and $A_{12,+}$, respectively. 
Obviously, after the collision of two solitons in the propagation process, their amplitudes will change. The amplitude of one soliton changes from $A_{11,-}$ to $A_{11,+}$, and the amplitude of the other soliton changes from $A_{12,-}$ to $A_{12,+}$. The $q_{2}$ component will have the same property. We assume that the corresponding amplitudes in the $q_{2}$ component are $A_{21,-}$, $A_{22,-}$, $A_{21,+}$, and $A_{22,+}$, respectively. Then we can obtain $A_{11,-}^2+A_{21,-}^2=A_{11,+}^2+A_{21,+}^2$ and $A_{12,-}^2+A_{22,-}^2=A_{12,+}^2+A_{22,+}^2$ from the analysis. This also means that the total mass of the components $q_{1}$ and $q_{2}$ is conserved before and after the collision.
	
	\subsection{Nondegenerate fractional soliton solution}
	In the above studies, we construct the fractional $n$-soliton solutions. When $n=1,2$, we list the explicit expressions in the equations \eqref{3-sol-1} and \eqref{3-sol-2}. Note that the bright soliton solutions we obtained in subsection \ref{subsection.one-soliton} are degenerate. Here, by degeneracy, we mean that the properties of the fundamental bright solitons in all components are described by a single wave number \cite{stalin2021nondegenerate}. The degeneracy will restrict the structure and trajectory of the single solitons. Thus the solitons in both $q_{1}$ component and $q_{2}$ component will only have a single-hump. We classify multi-component solitons into degenerate solitons and nondegenerate solitons according to the different wave numbers present in their solutions. We refer to the solitons that propagate with the same wave numbers in all components as degenerate vector solitons, and the solitons that propagate with different wave numbers as nondegenerate vector solitons \cite{stalin2021nondegenerate}. With a similar analysis as \cite{qin2019nondegenerate}, we find that the nondegenerate soliton solutions can be given by choosing the following special conditions,
	\begin{equation}\label{3-de-condition}
		\begin{cases}
			\gamma_{21}=\gamma_{32}=0\ \ {\rm or}\ \  \gamma_{22}=\gamma_{31}=0,\\[3pt]
			\big(a_{1}+2\alpha(3a_{1}^{2}-b_{1}^{2})\big)(a_{1}^{2}+b_{1}^{2})^{\epsilon}-\big(a_{2}+2\alpha(3a_{2}^{2}-b_{2}^{2})\big)(a_{2}^{2}+b_{2}^{2})^{\epsilon}=0.
		\end{cases}
	\end{equation}
	The two distinct complex 
	wave numbers $\lambda_{1}=a_{1}+\ii b_{1}$ and $\lambda_{2}=a_{2}+\ii b_{2}$ ($\lambda_{1}\neq\lambda_{2}$) in the equation \eqref{3-sol-2} makes the two-soliton solution 
	as nondegenerate one. 
	In this subsection, we take $\gamma_{21}=\gamma_{32}=0$ uniformly. 
	Then the nondegenerate fractional soliton solution is
	\begin{equation}\label{3-de-soliton}
		\begin{split}
			q_{1}^{[2]}=&\dfrac{4b_{2}\big((a_{1}-a_{2})+\ii(b_{1}+b_{2})\big)\frac{\gamma_{22}^{*}}{\gamma_{12}^{*}}\ee^{f_{4}^{*}}\Big((a_{1}-a_{2}-\ii b_{2})(1+\big|\frac{\gamma_{31}}{\gamma_{11}}\big|^{2}\ee^{2\Re(f_{3})})+\ii b_{1}(1-\big|\frac{\gamma_{31}}{\gamma_{11}}\big|^{2}\ee^{2\Re (f_{3})})\Big)}{\big((a_{1}-a_{2})^{2}+(b_{1}+b_{2})^{2}\big)\big(1+\big|\frac{\gamma_{31}}{\gamma_{11}}\big|^{2}\ee^{2\Re (f_{3})}\big)\big(1+\big|\frac{\gamma_{22}}{\gamma_{12}}\big|^{2}\ee^{2\Re(f_{4})}\big)-4b_{1}b_{2}},\\[4pt]
			q_{2}^{[2]}=&\dfrac{4b_{1}\big((a_{1}-a_{2})-\ii(b_{1}+b_{2})\big)\frac{\gamma_{31}^{*}}{\gamma_{11}^{*}}\ee^{f_{3}^{*}}\Big((a_{1}-a_{2}+\ii b_{1})(1+\big|\frac{\gamma_{22}}{\gamma_{12}}\big|^{2}\ee^{2\Re(f_{4})})-\ii b_{2}(1-\big|\frac{\gamma_{22}}{\gamma_{12}}\big|^{2}\ee^{2\Re (f_{4})})\Big)}{\big((a_{1}-a_{2})^{2}+(b_{1}+b_{2})^{2}\big)\big(1+\big|\frac{\gamma_{31}}{\gamma_{11}}\big|^{2}\ee^{2\Re (f_{3})}\big)\big(1+\big|\frac{\gamma_{22}}{\gamma_{12}}\big|^{2}\ee^{2\Re(f_{4})}\big)-4b_{1}b_{2}},
		\end{split}
	\end{equation}
	where
	\begin{multline}
		f_{k+2}=-2b_{k}x-4^{1+\epsilon}b_{k}\big(a_{k}+2\alpha(3a_{k}^{2}-b_{k}^{2})\big)(a_{k}^{2}+b_{k}^{2})^{\epsilon}t\\
		+\ii\Big(2 a_{k}x+2^{1+2\epsilon}\big(4\alpha a_{k}(a_{k}^{2}-3b_{k}^{2})+(a_{k}^{2}-b_{k}^{2})\big)(a_{k}^{2}+b_{k}^{2})^{\epsilon}t\Big),\ \ \ \ k=1,2.
	\end{multline}
	The solution \eqref{3-de-soliton} allows both the symmetric and the asymmetric profiles, including a double-hump and a single-hump profiles. The symmetric and asymmetric properties of the solution \eqref{3-de-soliton} can be determined by calculating the corresponding extremum points. 
	We assume $\gamma_{11}=\gamma_{12}=1$, then the different density profiles for double-hump solitons in fcHirota equation can be obtained by choosing appropriate values for the remaining parameters which are shown in Fig.\ref{fig:v-equal}. 
	The symmetric double-hump soliton in Fig.\ref{fig:v-equal}(b) is obtained numerically. According to the extreme value theory, we numerically find two maximum points of the modulu of $q_{1}^{[2]}$ in the equation \eqref{3-de-soliton}, and at these two maximum points, the corresponding maxima are equal. Under the selection of such parameters, symmetric double-hump soliton will appear.
	\begin{figure}
		\centering
		\includegraphics[width=0.9\linewidth]{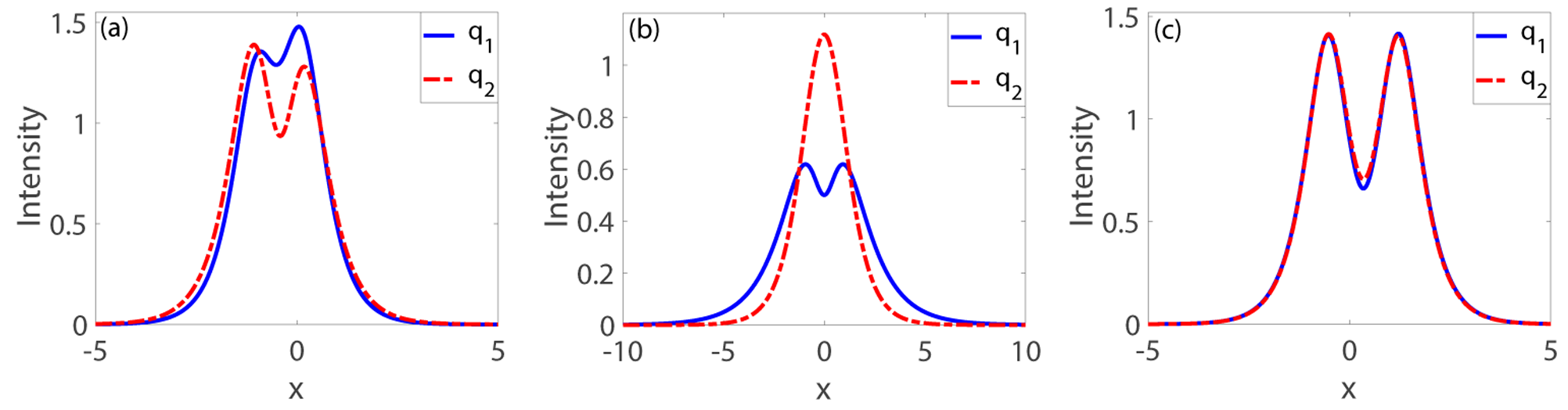}
		\caption{\small{Three different density profiles for double-hump solitons in fcHirota equation at $t=0$: (a) asymmetric double-hump solitons, (b) symmetric single-hump-double-hump solitons, and (c) approximate symmetric double-hump solitons. The blue solid lines and the red dotted lines represent the $q_{1}$ component and the $q_{2}$ component, respectively. The
				parameters are as follows: (a) $\epsilon=\frac{2}{5},\ \gamma_{31}=\gamma_{22}=\frac{1}{5},\ \alpha=\frac{1}{4},\ a_{1}=-\frac{9}{20},\ b_{1}=\frac{2\sqrt{5}}{5},\ a_{2}=-\frac{1}{20},\ b_{2}=1$, (b) $\epsilon=\frac{2}{5},\ \gamma_{31}=\gamma_{22}=\frac{99}{140}-\frac{1}{27800},\ \alpha=\frac{1}{4},\ a_{1}=0,\ b_{1}=\frac{5}{8},\ a_{2}=-\frac{1}{2},\ b_{2}=\frac{3}{8}$, (c) $\epsilon=\frac{2}{5},\ \gamma_{31}=\gamma_{22}=\frac{1}{2},\ \alpha=1,\ a_{1}=0,\ b_{1}=1,\ a_{2}=-\frac{1}{8},\ b_{2}=\frac{3\sqrt{7}}{8}$.}}
		\label{fig:v-equal}
	\end{figure}

	Furthermore, when we choose $\alpha=0$, the solution \eqref{3-de-soliton} is reduced to the nondegenerate fractional soliton solution of the fractional coupled NLS equation. And we display three different density profiles for double-hump solitons in Fig.\ref{fig:v-equal-cnls}. For convenience, we also assume $\gamma_{11}=\gamma_{12}=1$.
	\begin{figure}
		\centering
		\includegraphics[width=0.9\linewidth]{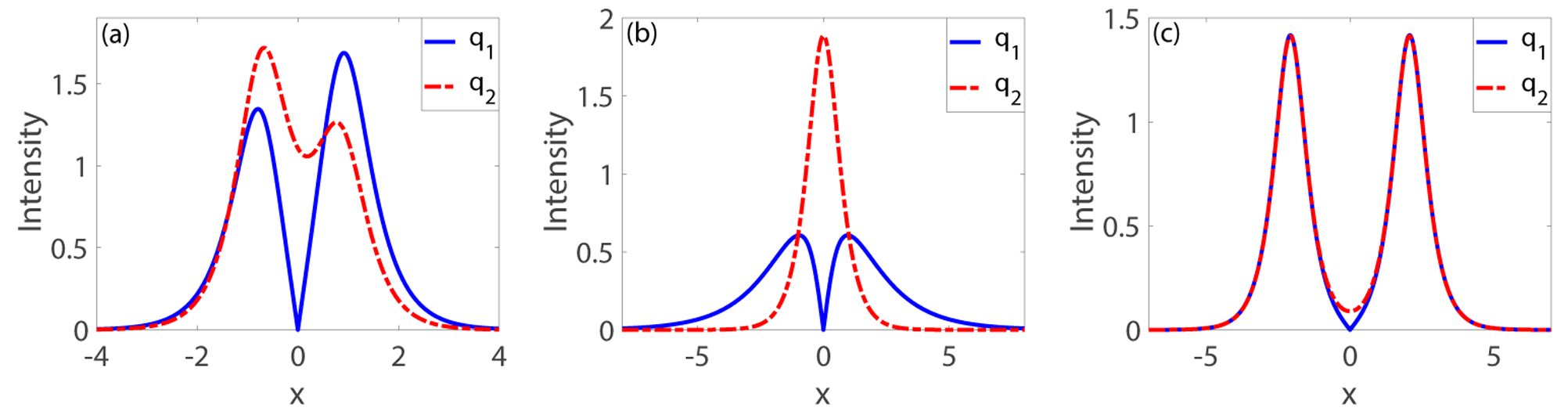}
		\caption{\small{Three different density profiles for double-hump solitons in fractional coupled NLS equation at $t=0$: (a) asymmetric double-hump solitons, (b) symmetric single-hump-double-hump solitons, and (c) approximate symmetric double-hump solitons. The blue solid lines and the red dotted lines represent the $q_{1}$ component and the $q_{2}$ component, respectively. Here the parameters
				are: (a) $\epsilon=\frac{2}{5},\ \gamma_{31}=\frac{1}{\sqrt{15}},\ \gamma_{22}=\frac{1}{3},\ a_{1}=a_{2}=0,\ b_{1}=\frac{8}{7},\  b_{2}=1$, (b) $\epsilon=\frac{2}{5},\  \gamma_{31}=\gamma_{22}=\frac{\sqrt{2}}{2},\ a_{1}=a_{2}=0,\ b_{1}=1,\  b_{2}=\frac{1}{3}$, (c) $\epsilon=\frac{2}{5},\  \gamma_{31}=\gamma_{22}=\frac{\sqrt{2001}}{2001},\ a_{1}=a_{2}=0,\ b_{1}=1+\frac{1}{1000},\  b_{2}=1$.}}
		\label{fig:v-equal-cnls}
	\end{figure}
	
	Note that when $t=0$, the solution \eqref{3-de-soliton} appears to be  independent of the parameter $\alpha$, but in fact, the parameters $a_{1},\ a_{2},\ b_{1}$, and $b_{2}$ must satisfy the second condition in the equation \eqref{3-de-condition}. Therefore, the solution \eqref{3-de-soliton} at $t=0$ will still be controlled by the parameter $\alpha$.
	Moreover, it is worth noting that the solution \eqref{3-de-soliton} appears zero if and only if $\alpha=0$ and $b_{1}>b_{2}$. For example, it is obvious that the parameters taken in Fig.\ref{fig:v-equal-cnls} (a)-(c) satisfy the above conditions.

	\section{Conclusion}
	In conclusion, we developed the fractional extension of the coupled Hirota equation in terms of the Riesz fractional derivative. Then we investigate the IST with the matrix RHP to study the fractional $n$-soliton solutions of the fcHirota equation. Furthermore, we analyzed the one- and two-soliton solutions of this equation. It can be found that the physical properties of the fractional solitons are identical to the regular ones expect the change in velocities described by equation \eqref{3-velocity}. And from Figs.\ref{fig:sol-1} and \ref{fig:sol-2}, the fractional solitons propagate without dissipating
	or spreading out. Moreover, we showed that a fractional two-soliton could also be
	regarded as a superposition of two fractional single solitons of the fcHirota equation as $|t|\to\infty$. In particular, we also study the nondegenerate fractional soliton solutions. In addition, the general formula of the fractional $n$-soliton solution of the fcHirota equation is derived. 
	These phenomena have
	greatly enriched the dynamics in the field of integrable systems, and will be helpful in predicting the super-dispersive transport of nonlinear waves in
	fractional nonlinear media.
	
	Furthermore, there is an important equation in integrable systems, which is called the Sasa-Satsuma equation. It can also be regarded as the higher-order NLS equation, which has the form
	\begin{equation}\label{3-HNLS}
		\ii q_{t}+\alpha_{1}q_{xx}+\alpha_{2}|q|^{2}q+\ii\delta\big(\beta_{1}q_{xxx}+\beta_{2}|q|^{2}q_{x}+\beta_{3}q(|q|^{2})_{x}\big)=0,
	\end{equation}
	where $\alpha_{j},\ \beta_{j}$ are all real constants, and $q(x,t)$ is a complex valued function. When $\beta_{1}:\beta_{2}:\beta_{3}=1:6:0$ and $\beta_{1}:\beta_{2}:\beta_{3}=1:6:3$, the equation \eqref{3-HNLS} can be reduced to the Hirota equation \cite{hirota1973exact} and the Sasa-Satsuma equation \cite{sasa1991new}, respectively. Obviously, the Sasa-Satsuma equation has one more term called the stimulated Raman scattering than the Hirota equation. In \cite{xu2014soliton,nimmo2015binary,liu2017certain}, it can be seen that the properties of the corresponding solutions of the Sasa-Satsuma equation are more complicated than the counterparts of the Hirota equation. So it is significant to study the Sasa-Satsuma equation. Compared with the Hirota equation, the Sasa-Satsuma equation also admits the symmetric double-hump solitons in addition to the common single-hump solitons under the vanishing background. The double-hump solitons exhibit like two in-phase solitons propagating at the same speed and a fixed separation. However, for the Hirota equation, the double-hump solitons appear only in its two-component or multi-component system. In Sec.\ref{3-sec.3}, we find that the double-hump solitons also exist in fcHirota equation, and the related properties of the solitons of the fcHirota equation are very similar to the coupled Hirota equation. We can speculate that the fractional Sasa-Satsuma equation may also have the same properties as the Sasa-Satsuma equation, for example, the symmetric double-hump solitons may also exist under the vanishing background. All these issues discussed above are worth our further consideration. 
	
	\section*{Acknowledgements}
	Liming Ling is supported by National Natural Science Foundation of China (No. 12122105). Xiaoen Zhang is supported by the National Natural Science Foundation of China (No.12101246), the China Postdoctoral Science Foundation (No. 2020M682692), the Guangzhou Science and Technology Program of China(No. 202102020783).
	
	\section*{Conflict of interest}
	The authors declare that they have no conflict of interest with other people or organization that may inappropriately influence the author's work.

	\bibliographystyle{unsrt}
	
	\bibliography{reference}
	
\end{document}